\documentclass[10pt]{article}
\usepackage{boxedminipage}
\usepackage{verbatim}
\usepackage{multirow}
\usepackage{fullpage}
\usepackage{subfig}

\usepackage[english]{babel}
\usepackage{amsthm, amsmath, amssymb}
\usepackage{algorithm, algorithmic}
\usepackage{color}
\usepackage{mathtools}

\DeclareMathOperator{\E}{\mathbb{E}}
\DeclareMathOperator{\negl}{\text{negl}}

\newtheorem{theorem}{Theorem}[section]

\newtheorem{lemma}[theorem]{Lemma}
\theoremstyle{definition}
\newtheorem{definition}{Definition}[section]

\newcommand{\ignore}[1]{}
\newcommand{\zu}{\{0,1\}}                           

\newcommand{\calA}{{\mathcal{A}}}
\newcommand{\calC}{{\mathcal{C}}}

\newcommand{\calT}{{\mathcal{T}}}
\newcommand{\B}{{\mathbb{B}}}

\newcommand{\key}{{\mathtt{key}}}

\newcommand{\SetupA}{Setup}
\newcommand{\ShuffA}{Shuffling}
\newcommand{\ObShuffA}{Oblivious \ShuffA}
\newcommand{\KObShuffA}{$K$-Oblivious \ShuffA}

\newcommand{\Shuffle}{{{\mathsf{Sh}}}}              
\newcommand{\ShuffG}{\mathsf{OSGame}}               
\newcommand{\transD}{{\calT}}                       
\newcommand{\trans}{{\mathtt{trans}}}               
\newcommand{\transmD}{{\mathcal{M}}}                
\newcommand{\transm}{{\mathtt{Mtrans}}}             

\newcommand{\touched}{{\mathsf{Touched}}}           
\newcommand{\K}{\touched}                           

\newcommand{\nameAlgoS}{{\sf CacheShuffle}}         
\newcommand{\firstC}{{\sf K}\nameAlgoS{\sf Basic}}  
\newcommand{\secondC}{\nameAlgoS{\sf Root}}         
\newcommand{\thirdC}{\nameAlgoS}                    

\newcommand{\fullC}{\nameAlgoS{\sf Root}}
\newcommand{\kC}{{\sf K}\nameAlgoS}
\newcommand{\ksecondC}{{\sf K}\nameAlgoS{\sf Root}}

\newcommand{\Spray}{\text{\sf Spray}}               
\newcommand{\Recalibrate}{\text{\sf Recalibrate}}   
\newcommand{\RSpray}{\text{\sf RSpray}}             

\newcommand{\source}{{\mathsf{Source}}}
\newcommand{\inputR}{{\mathsf{R}\source}}	  
\newcommand{\inputN}{n}                           
\newcommand{\destInd}{D}                          
\newcommand{\destIndN}{d}                         
\newcommand{\temp}{{\mathsf{temp}}}               

\newcommand{\blockSub}{{\mathsf{srcInd}}}         
\newcommand{\sourceB}{{\mathsf{sourceB}}}
\newcommand{\Aa}{{\source}}  			   
\newcommand{\dest}{{\mathsf{Dest}}}
\newcommand{\Dd}{{\dest}}                          
\newcommand{\dBuck}{{\mathsf{destB}}}              
\newcommand{\D}{\dBuck}                            
\newcommand{\T}{{\temp}}                           
\newcommand{\rem}{{\mathsf{Rem}}}                  

\newcommand{\Q}{{\mathsf{Q}}}

\newcommand{\A}{{\blockSub}}                        
\newcommand{\ntBucket}{\nu}                     
\newcommand{\ndBucket}{{\xi}}            
\newcommand{\tbDown}{{\mathsf{tbDown}}}

\newcommand{\dummyAlgoS}{{\sf KCacheShuffleDummy}}
\newcommand{\indCpa}{{\it IND-CPA}}

\newcommand{\kO}{\KObShuffA}

\newcommand{\kOS}{{\KObShuffA}}

\DeclareMathOperator{\poly}{\text{poly}}

\DeclareMathOperator{\Enc}{\text{\sf Enc}}
\DeclareMathOperator{\Dec}{\text{\sf Dec}}

\newcommand{\fourthC}{{\dummyAlgoS}}

\newcommand{\dMap}{{\mathsf{dMap}}}
\newcommand{\bZero}{{\mathbf{0}}}
\newcommand{\cnt}{{\mathsf{cnt}}}
\newcommand{\pos}{{\mathsf{pos}}}

\newcommand{\idx}{{\mathsf{idx}}}
\newcommand{\F}{{\mathbb{F}}}

\newcommand{\Key}{{\key}}

\newcommand{\U}{{\mathbb{U}}}
\newcommand{\adv}{{\calA}}
\newcommand{\ch}{{\calC}}

\newcommand{\from}{\leftarrow}
\newcommand{\osGame}{{\mathsf{OSGame}}}
\newcommand{\shuffleInd}{{\mathsf{ShuffleIndGame}}}
\newcommand{\indCpaGame}{{\mathsf{IndCPAGame}}}

\newcommand{\tCt}{\mathsf{touchCt}}
\newcommand{\tInd}{\mathsf{touchInd}}
\newcommand{\dCt}{\mathsf{destCt}}
\newcommand{\dInd}{\mathsf{destInd}}

\title{CacheShuffle: An Oblivious Shuffle Algorithm using Caches\thanks{
The first version of this report (May 19, 2017) described $\nameAlgoS$ and
$\secondC$.
The second version of this report (September 5, 2017) introduced the concept of a
{\KObShuffA} and described $\firstC$ and $\dummyAlgoS$.
The current version describes $\kC$.
}}
\author{
Sarvar Patel\thanks{Google, Inc., {\tt sarvar@google.com}}\and
Giuseppe Persiano\thanks{Google, Inc. and Universit\`a di Salerno, {\tt giuper@gmail.com}}\and
Kevin Yeo\thanks{Google, Inc., {\tt kwlyeo@google.com}}
}
\date{\today}

\begin{document}
\maketitle

\begin{abstract}

We consider the problem of {\em\ObShuffA}, a critical component in several applications
in which one wishes to hide the pattern of data access,
and the problem of {\em\KObShuffA}, a refinement thereof.
We provide efficient algorithms for both problems and discuss their application to
the design of Oblivious RAM.
The task of a {\KObShuffA} algorithm is to obliviously shuffle
$N$ encrypted blocks that have been randomly allocated on the server
in such a way that an adversarial server learns nothing about the new allocation of blocks.
The security guarantee should hold also with respect to an adversary that
has learned the initial position of $K$ {\em touched} blocks out of the $N$ blocks.
The classical notion of {\ObShuffA} is obtained for $K=N$.

\smallskip
We start by presenting a family of algorithms for {\ObShuffA}.
Our first construction, that we call $\fullC$, is tailored for clients with $O(\sqrt{N})$ blocks of
memory and uses $(4+\epsilon)N$ blocks of bandwidth, for every $\epsilon>0$ and has negligible in $N$
failure probability.
$\fullC$ is a 4.5x improvement over the previous best known result on practical sizes of $N$.
We also present $\thirdC$ that obliviously shuffles using $O(S)$ blocks of client
memory with $O(N\log_S N)$ blocks of bandwidth.

\smallskip
We then turn to {\KObShuffA} and give
algorithms that require $2N + f(K)$ blocks of bandwidth,
for some function $f$.
That is, any extra bandwidth above the $2N$ lower bound depends solely on $K$.
We present $\firstC$ that uses
$O(K)$ client storage and exactly $2N$ blocks of bandwidth.
For smaller client storage requirements, we show
$\kC$, which uses $O(S)$ client storage and requires
$2N + (1+\epsilon)O(K \log_S K)$ blocks of bandwidth.

\smallskip
Finally, motivated by applications to ORAM design, we consider also the case in which,
in addition to the $N$ blocks, the server stores $D$ dummy blocks whose content is is irrelevant but
still their positions must be hidden by the shuffling.
For this case, we design algorithm \dummyAlgoS\ that,
for $N+D$ blocks and $K$ touched blocks, uses $O(K)$ client storage and
$D+(2+\epsilon)N$ blocks of bandwidth.

\smallskip
We discuss how to use \firstC\ and \dummyAlgoS\ to improve practical Oblivious RAM constructions.
\end{abstract}

\clearpage

\tableofcontents
\clearpage

\section{Introduction}
Cloud storage has become an increasingly popular technology due to the many benefits it offers users.
Uploading files to the cloud allows users to share documents easily without incurring into
bandwidth costs or the annoyance of email attachments.
Users are able to access documents from anywhere without having
to transfer data between machines. The burden of replicating data and
recovering from machine failures is placed on the storage provider.
For many corporations, cloud storage becomes cost efficient since the price of cloud storage
may be cheaper than developing and maintaining their own internal storage
systems.

Some users might want to hide the contents of their data from their cloud providers.
A first attempt would be to encrypt all documents by the client before uploading the files to the server.
Work done in \cite{Liu:2014:SPL:2580107.2580271} and
\cite{Naveed:2015:IAP:2810103.2813651} show that the access pattern
to encrypted data may leak information to cloud storage providers.
Data oblivious algorithms and storage can be used to hide the access pattern
to remote data with Oblivious Random Access Memory (ORAM) being the primary example.
ORAM was first introduced by \cite{Goldreich:1987:TTS:28395.28416}
(see also~\cite{Goldreich:1996:SPS:233551.233553}) that introduced the so called Square Root ORAM
(with a square root communication overhead and client memory)
and the Hierarchical ORAM construction which has poly-logarithmic amortized
cost and sublinear client storage. Recently, more practical constructions
such as Path ORAM \cite{Stefanov:2013:POE:2508859.2516660},
Partition ORAM \cite{partition_oram} and
Recursive Square Root ORAM \cite{cryptoeprint:2017:964} have appeared.
The best known asymptotic results with $O(\log N)$ amortized query cost with $O(N^\epsilon)$ blocks
of client storage was first shown in \cite{Goodrich:2012:PGD:2095116.2095130}.
The result was improved to make the worst case also $O(\log N)$ in
\cite{Goodrich:2011:ORS:2046660.2046680}. However, these asymptotic results
have constants too large to be considered practical at the moment.

Many ORAM constructions have amortized costs due to the need of periodically running
an {\em oblivious shuffling} algorithm. Roughly speaking, an oblivious shuffle
moves around the data blocks in such a way that the server cannot link together
the position of a block before the shuffle
with the position of the same block after the shuffle.
This is used to completely remove any links that the server might have created
with the data blocks in the position before the oblivious shuffle.
Essentially, all extracted information is rendered useless.
This idea appeared in the original Square Root ORAM and Hierarchical ORAM that
allowed clients to perform accesses until a break point was reached where the server might
be able to extract extra information from the access pattern.
At this point, the client performs an oblivious shuffle (after which, new queries cannot provide
extra information to the server), and queries can occur again.

The early approach to oblivious shuffling involved the use of sorting circuits
(or of oblivious sorting algorithms that can also be seen as sorting circuits).
The client evaluates the compare-exchange gates one at a time and for each gate
it downloads the two encrypted blocks that are input to the gate,
re-encrypts them and uploads them in right order.
Batcher's sort is considered the most practical algorithm
\cite{Batcher} even though it has asymptotic cost of $O(N \log^2 N)$.
Sorting networks such as AKS~\cite{AKS} and Zig-Zag~\cite{Goodrich:2014:ZSS:2591796.2591830} have $O(N \log N)$ size, but
are considered impractical due to large constants.
Randomized Shellsort~\cite{Goodrich:2011:RSS:2049697.2049701}
has larger depth than AKS but the constant hidden
in the big Oh notation is smaller.
Oblivious shuffling based on sorting circuits is interesting because the client need only to store
a constant number of data blocks but it requires bandwidth proportional to the size of the network which
is $\Omega(N\log N)$.
The first oblivious shuffling algorithm not based on sorting circuits, the {\em Melbourne Shuffle},
was introduced in \cite{Ohrimenko2014} and uses $O(N)$ bandwidth
while only requiring $O(\sqrt N)$ blocks to be stored on the client at any time.

\paragraph{Results and Contributions.}
In this paper, we present practical algorithms for oblivious shuffling.

We base our design on the following observation that has been overlooked by previous work.
As we have discussed, the main goal of oblivious shuffling is to make sure that the
adversary cannot accumulate too much information on which slot in server memory holds which block
in algorithms that hide the access pattern to data blocks.
However, it is seldom the case that the adversary gets to learn the position of {\em all} the $N$ blocks
but, rather, of a number of blocks that is equal to the size of the client memory.
In addition, the client knows exactly which blocks have been touched by the server.
We call these blocks the {\em touched} blocks.
This is the case, for example,
for the Square Root ORAM of \cite{Goldreich:1987:TTS:28395.28416,Goldreich:1996:SPS:233551.233553} and of
its hierarchical versions.
Motivated by this observation, we introduce the
concept of a {\em {\KObShuffA} Algorithm} that
reshuffles $N$ data blocks, $K$ of which are touched. For $K=N$,
the notions of a {\KObShuffA} Algorithm coincides with the original notion of an
{\ObShuffA} Algorithm of~\cite{Ohrimenko2014}.

Before tackling the problem of designing efficient {\KObShuffA} Algorithms,
we revisit the original Oblivious Shuffling problem by providing improved algorithms.
All our algorithms use a {\em cache} in client memory to store blocks downloaded from the server before
they can be uploaded to the server. The main technical difficulty is to show that the cache does not
grow too much.
We first apply this design principle in Section~\ref{sec:fullRoot} by presenting an {\ObShuffA} Algorithm,
$\fullC$,
that uses bandwidth $(4+\epsilon)N$ and client memory $O(\sqrt{N})$ and has negligible failure probability.
For similar client memory usage and error probability, Melbourne Shuffle~\cite{Ohrimenko2014}
uses about 4 times more bandwidth.
We generalize this construction to $\thirdC$ in Section~\ref{sec:fullRecursive}
to any client memory $S=\omega(\log N)$ in which case bandwidth is $O(N\log N/\log S)$ and
probability of failure is still negligible.

We then turn to the design of {{\KObShuffA} Algorithms} for $K\leq N$.
From a high level,
the number $K$ of touched data blocks succinctly describes the difficulty of shuffling the specific data sets.
In the extreme case that no block has been disclosed
(e.g., in an ORAM in which no query has been performed),
clearly no shuffle is required. On the other hand, if all blocks have been touched,
then oblivious shuffling becomes hardest.
All previous oblivious shuffling algorithms have always considered the most difficult scenario and have
reshuffled so to guarantee security as if all the blocks had been touched, even if that was not the case
in the specific application.
Our work is the first to separate the two problems.
In Section~\ref{sec:smallK}, we give a simple {{\KObShuffA} Algorithm}
$\firstC$
for the case in which client memory $S\geq K$. This setting is relevant to ORAM design and,
for every $K$, we obtain an algorithm with bandwidth $2N$.
In Section~\ref{sec:largeK}, we give, for every client memory $S=\omega(\log N)$ and for every $\epsilon$,
algorithm $\kC$ that uses bandwidth $2N+c\cdot(1+\epsilon)K\log_S K$, for constant $c$.
For the special case of $S=\sqrt{K}$, we have algorithm $\ksecondC$ that uses bandwidth $2N+(4+\epsilon) K$.
For every $\epsilon>0$, the algorithms have negligible in $N$ abort probability.

Motivated by the problem of designing bandwidth efficient ORAM, in Section~\ref{sec:dummy} we consider a scenario in which
there are $D$ {\em dummy} blocks, which can be of arbitrary values, in addition to $N$ {\em real} blocks.
It is possible to use any {{\KObShuffA} Algorithm}
in this scenario and just treat the dummy blocks as any other block.
By taking into account instead the fact that the content of the dummy blocks is irrelevant
we present algorithm \dummyAlgoS\ that has bandwidth of $D+2(N+\epsilon)$ blocks for some small $\epsilon>0$
for the case in which $S\geq K$. Applying directly $\firstC$ would result in bandwidth $2(N+D)$ and so we
are saving bandwidth corresponding to $D$ blocks. The savings come at the cost of
a small amount of server computation,

In the table below,
we compare our algorithms with the previous best algorithm for {\ObShuffA},
Melbourne Shuffle~\cite{Ohrimenko2014}.

\begin{figure}[H]
\begin{center}
        \begin{tabular}{ | l | l | l |}
        \hline
        &  Client Storage & Bandwidth\\ \hline
	Melbourne Shuffle~\cite{Ohrimenko2014} & $O(\sqrt{N})$ & $\approx 18N$\\ \hline \hline
	\secondC  & $O(\sqrt{N})$ & $(4+\epsilon)N$\\ \hline
	\thirdC   & $O(S)$        & $O(N\log_S N)$\\ \hline
        \firstC   & $O(K)$        & $2N$\\ \hline
	\ksecondC & $O(\sqrt{K})$ & $2N+(4+\epsilon)K$\\ \hline
	\kC       & $O(S)$        & $2N+(1+\epsilon)O(K\log_S K)$\\ \hline
	\fourthC  & $O(K)$        & $D+(2+\epsilon)N$\\ \hline
        \end{tabular}
\end{center}
\caption{$N$ denotes the number of blocks.
Algorithm \fourthC\ receives $D$ additional dummy blocks, for a total of $N+D$ blocks.
Algorithm \ksecondC\ is obtained from algorithm \kC\ by setting $S=\sqrt{N}$.
For all algorithms, server storage is $c\cdot N$, for small constant $c$.
}
\end{figure}

\section{Definitions}
Our reference scenario is a cloud storage model with a {\em client} that
wishes to outsource the storage of $N$ data blocks of identical sizes to a {\em server} that has storage of capacity $M\geq N$.
In this context, we consider the problem of obliviously shuffling the data blocks.

We assume that the data blocks have been uploaded by the {\em\SetupA} algorithm that takes as input a sequence
$\B=(B_1,\ldots,B_N)$ of $N$ data blocks of identical sizes $B$
and a {\em permutation} $\pi:[N]\rightarrow [N]$.
The \SetupA\ algorithm randomly selects an encryption key $\key$ for a symmetric encryption scheme
and uploads the data blocks encrypted using $\key$ to the server
by storing it in the first $N$ locations of the server storage according to $\pi$;
that is, if $\pi(i)=j\in[N]$, an encrypted copy of the $i$-th data block is stored at the $j$-th location of
the server storage.
Note that $\pi$ is a permutation and each of the $N$ data blocks is uploaded exactly once to the server.

Once the data has been uploaded, an adversary $\adv$ is allowed to query for the position
of a subset $\touched$ of the data blocks and for each queried block $i$,
the value $\pi(i)$ is revealed to $\adv$.
We call the data blocks in $\touched$, the {\em touched} data blocks.

The \ShuffA\ algorithm, instead, takes as input the encryption key $\key$ used to setup the
$N$ blocks, the permutation map $\pi$, the set $\touched$ of touched data blocks and a new permutation $\sigma$.
The task of the {\ShuffA} algorithm is to re-permute the $N$ data blocks stored on the server
according to permutation map $\sigma$.
At each step, the \ShuffA\ algorithm can download a block $i$
to {\em client memory} by specifying the block's current location on the server
or upload a block from client memory to server memory by specifying its destination on the server.
In addition, the \ShuffA\ algorithm can ask the server to perform operation on locally stored data blocks.
We are interested in {\em oblivious} \ShuffA\ algorithms that, roughly speaking,
have the property of hiding information about the content of the data blocks and on $\sigma$,
even to an adversarial algorithm that has partial information on $\pi$ (the set $\touched$)
and observes the blocks downloaded and uploaded by the \ShuffA\ algorithm.

\paragraph{The mechanics of the \ShuffA\ algorithm.}
A \ShuffA\ algorithm receives as input the initial permutation $\pi$, the final permutation $\pi$ and the set $\touched$.
The \ShuffA\ algorithm proceeds in steps and each step can be either
a {\em move} step or a {\em server computation} step.
The state after the $q$-th step is described by a {\em server allocation map}
$\rho_q:[M]\rightarrow[N]\cup\{\perp\}$ and by a {\em client allocation map} $L_q:[S]\rightarrow[N]\cup\{\perp\}$.
Each allocation map specify the block in each of the $M$ server locations and
$S$ client locations, respectively.
More precisely, $\rho_{q}(j)=i$ means that, after the $q$-th step is performed,
the $j$-th server location contains an encryption of the $i$-th data block.
If instead $\rho_{q}(j)=\perp$, then an encryption of a dummy data block is stored at location $j$.
Note that, unlike permutations, the argument of an allocation map is an index of a slot in memory and its value is a block index.
Similarly statements are true for the client allocation map but $L_q(j)=\perp$ denotes an empty client slot.

When a \ShuffA\ algorithm starts,
the server allocation map $\rho_0$ coincides with permutation map $\pi$
on the first $N$ storage location of the $M$ server
memory locations and has dummy blocks on the remaining $N-M$ locations;
that is, for $1\leq i\leq N$, $\rho_0(i)=\pi^{-1}(i)$ and $\rho_0(i)=\perp$, for $i>N$.
Instead $L_0[i]=\perp$ for all $i\in[S]$ (that is, initially, no block is stored in the client's local memory).
At each step, a \ShuffA\ algorithm can perform
either a {\em move} operation or a {\em server computation} operation.
A move operation can be either a {\em download} or an {\em upload} move and they
modify the state as follows.
If the $q$-th move is a download move with {\em source} $s_q$ and {\em destination} $d_q$,
it has the effect of storing an encryption  of block $\rho_{q-1}(s_q)$ stored at server location $s_q$
at location $L_q(d_q)$ of the client memory; that is, the block at location $s_q$
on the server is downloaded, decrypted using $\key$ and re-encrypted by using $\key$ and fresh randomness.
As a consequence, the server allocation map stays the same $\rho_q:=\rho_{q-1}$
and $L_q$ is the same as $L_{q-1}$ with the exception that
$L_q(d_q)=\rho_{q-1}(s_q)$.
If instead, the $q$-th move is an upload move with {\em source} $s_q$ and {\em destination} $d_q$,
it has the effect of uploading the block in client location $s_l$ to server location $d_q$;
that is, the client allocation map stays the same $S_q:=S_{q-1}$ and $\rho_{q}$ differs from
$\rho_{q-1}$ only for the values at $d_q$.
Our algorithms will also use special upload moves with $s_q=\perp$
in which a dummy block (say, a block consisting of all $0$'s) is uploaded to server location $d_q$.
A server computation operation is instead specified by a circuit that takes as input
a subset of the blocks and modifies the blocks stored at the server. As we shall see, this operation
consists of homomorphic operation on ciphertext and can be used to save bandwidth while requiring more
server computation.
can also perform operations on blocks stored on the server in which case
the length of the description of the circuit describing the operation is added to the bandwidth.
We compute the {\em bandwidth} of a \ShuffA\ algorithm using the block size $B$ as a unit of measurement;
thus bandwidth is simply the number of move operations plus
the size of the circuits corresponding to server computation operations divided by $B$.

\paragraph{Efficiency measures.}
Three measures of efficiency can be considered for a \ShuffA\ algorithm:
the total bandwidth $T$,
the amount $S$ of client memory and the amount $M$ of server memory.
Note that oblivious shuffling of $N$ data blocks is trivial for clients with memory $S\geq N$:
download all the $N$ encrypted blocks in some fixed order; decrypt and re-encrypt each block;
finally, upload the newly encrypted blocks to the new locations one by one in some fixed order.

In this paper,
we give oblivious shuffling algorithms that use memory $S=o(N)$ and server memory $M=O(N)$.
In most cases, server memory is cheaper than block transfers,
so we do not try to optimize for the hidden constants of server memory
(which is however small for all our constructions).
Our main objective is to design algorithms with small $T$.

\paragraph{The security notion.}
A {\em transcript} of an execution of a \ShuffA\ algorithm $\Shuffle$
consists of the initial content of the server memory,
the ordered list of the sources of all download moves,
the ordered list of the destinations of all the upload moves as well as the data blocks uploaded with each
move,
and the list of circuits uploaded by the client.
We stress that a transcript only contains the server locations that are involved in each move
(that is the source for the downloads and the destination for the uploads)
but not the client locations so to model the fact that an adversarial server $\adv$
cannot observe where each block is stored when downloaded
and from which client location each uploaded block comes from.

For every sequence of $N$ blocks $\B=(B_1,\ldots,B_N)$, every subset $\touched$ of touched blocks,
and every pair of permutations $(\pi,\sigma)$,
a \ShuffA\ algorithm $\Shuffle$ naturally induces a probability distribution
$\transD_{\Shuffle}(\B,\pi,\sigma,\touched)$ over all possible transcripts.
We capture the notion of a {\em \KObShuffA} algorithm
by the following game $\ShuffG^\adv_\Shuffle$ for \ShuffA\ algorithm $\Shuffle$
between an adversary $\adv$ and the challenger $\ch$.
In the formalization of our notion of security, we allow the adversary $\adv$ to
receive partial information on the starting permutation map $\pi$ to
reflect the fact that the \ShuffA\ algorithm $\Shuffle$ might be part of a larger
protocol whose execution leaks information on $\pi$. More precisely, in our formalization
we allow $\adv$ to choose the initial location on the server of a subset $\touched$ of the $N$ data blocks
and we parametrize the security notion by the cardinality $K$ of the set $\touched$.
The challenger $\ch$ fills in the remaining $N-K$ locations randomly under the constraint that each of the
$N$ blocks appears in exactly one location on the server.
Then, $\adv$ proposes
two sequences, $\B_0$ and $\B_1$, of $N$ blocks
and
two permutations, $\sigma_0$ and $\sigma_1$, and $\ch$ randomly picks $b\in\zu$
and samples a transcript $\transmD^{\Shuffle}(\Enc(K, \B_b),\pi,\sigma_b,\touched)$ according to $\transD^{\Shuffle}(\Enc(K,\B_b),\pi,\sigma_b,\touched)$.
$\adv$ then, on input $\trans$, outputs its guess $b'$ for $b$.
We say that an adversary $\adv$ is {\em $K$-restricted} if
it specifies the location of at most $K$ blocks; that is $|\touched|\leq K$.

\begin{definition}
For shuffle algorithm $\Shuffle$ and adversary $\adv$,
we define game
$\ShuffG^\adv_{\Shuffle}(N,\lambda)$ as follows
\begin{enumerate}
	\item $\adv$ chooses a subset $\touched\subseteq[N]$ and specifies $\pi(i)$ for each $i\in\touched$;
	\item
	$\adv$ chooses two pairs $(\B_0, \sigma_0)$ and $(\B_1, \sigma_1)$
	and sends them to $\ch$;
\item $\ch$ completes the permutation $\pi$ by randomly choosing the values
at the point left unspecified by $\adv$;
\item $\ch$ randomly selects $b\from\{0,1\}$ and sends $\adv$  transcript $\transmD^{\Shuffle}(\Enc(K,\B_b),\pi,\sigma_b,\touched)$,
	drawn according to $\transD^{\Shuffle}(\Enc(K,\B_b),\pi,\sigma_b,\touched)$;
\item $\adv$ on input $\transmD^\Shuffle(\Enc(K,\B_b),\pi,\sigma_b,\touched)$ outputs $b'$;
\end{enumerate}
Output $1$ iff $b=b'$.
\end{definition}

\begin{definition}[{\KObShuffA}]
We say that \ShuffA\ algorithm $\Shuffle$ is a {\em \KObShuffA} algorithm if for all
$K$-restricted probabilistically polynomial time adversaries $\adv$, and for all $N=\poly(\lambda)$
	$$\Pr[\ShuffG^\adv_\Shuffle(N,\lambda)=1]\le\frac{1}{2}+\negl(\lambda).$$
\end{definition}
\noindent
We refer to $N$-Oblivious Shuffling as just Oblivious Shuffling.

\section{Tools}
In this section we review some of the tools we use to prove security of our constructions.
\subsection{Encryption}\label{ind-cpa}
As we previously mentioned, the basic operation of oblivious shuffling
involves either download a block from the server to the client
or uploading a block from the client to the server.
This means that the same block could be downloaded in one step
and subsequently uploaded in a future step.
We wish to prevent the server from linking that the same block was
being downloaded/uploaded at various steps.

To prevent the server from linking data contents, the client can always
decrypt and encrypt each data block with new randomness that is independent
on the input and output permutations. The \indCpa\ game encompasses
the above needs. In simple terms, \indCpa\ states that the encryption
of two
plaintexts are indistinguishable.

\begin{definition}[IND-CPA]
Let $\calA$ be an adversary and $\calC$ consisting of $\Enc$ and $\Dec$ be
the challenger. The following game between $\calA$ and $\calC$ is defined
as the $\indCpaGame^{\calA}_{(\Enc, \Dec)}(\lambda)$ game.

\begin{enumerate}
	\item
	$\calC$ generates private key $\key$ of length $\lambda$;
	\item
	$\calA$ asks for $\poly(\lambda)$ encryptions under $\key$ from $\calC$;
	\item
		$\calA$ submits two distinct plaintexts $p_0$ and $p_1$
		as the challenge;
	\item
		$\calC$ picks secret bit $b$ and sends
		$\Enc(\key, p_b)$
		to $\calA$;
	\item
	$\calA$ asks for $\poly(\lambda)$ encryptions under $\key$ from $\calC$;
	\item
	$\calA$ outputs $b'$;
\end{enumerate}
Output 1 iff $b = b'$.
\end{definition}

\begin{definition}[IND-CPA secure]
We say that the encryption scheme $(\Enc, \Dec)$ is IND-CPA secure if
for all probabilistically-polynomial time adversaries $\calA$,
$$
	\Pr[\indCpaGame^{\calA}_{(\Enc, \Dec)}(\lambda) = 1] \le \frac{1}{2} + \negl(\lambda).
$$
\end{definition}
\noindent
Throughout the rest of this work, we will assume that $(\Enc,\Dec)$ is secure under
\indCpa.

\subsection{Pseudorandom Permutations}
In the problem definition, we state that the input of the {\em Shuffle} problem
includes two permutations, $\pi$ and $\sigma$. In general, storing
true random permutations requires $O(N \log N)$ bits via information theory
lower bounds. However, it is possible to have space-efficient
constructions for pseudorandom permutations. Furthermore, we still wish for
the permutation to be accessible, that is fast to evaluate $\pi(i)$ for any $i$.
For example, we do not want to be required to use $O(N)$ computation to find $\pi(i)$.

One of the first space-efficient pseudorandom permutations was by Black and Rogaway~\cite{Black2000}, which required the storage of only three keys. However,
their scheme only provided security guarantees for up $N^{1/4}$ evaluations.
Work by Morris et al~\cite{crypto-2009-23861} pushed the guarantees up to
$N^{1-\epsilon}$ queries. The construction by Hoang et al~\cite{Hoang2012}
pushed security up to $(1 - \epsilon)N$ queries until the Mix-and-Cut Shuffle~\cite{Ristenpart2013} provided a fully-secure pseudorandom permutation allowing
evaluation on all $N$ possible inputs. The Sometimes-Recurse Shuffle~\cite{Morris2014} the efficiency of the Mix-and-Cut Shuffle allowing evaluations in $O(\log N)$ AES evaluations while only storing a single key.

For any sublinear storage Oblivious Shuffling algorithms to make sense, we will
assume that the input and output permutations $\pi$ and $\sigma$ are pseudorandom permutations with small storage. In practice, the Sometimes-Recurse Shuffle~\cite{Morris2014} would suffice.

\subsection{Proving $K$-Obliviousness for Move-Based \ShuffA\ Algorithms}
Move-based algorithms only perform move operations between the server storage and the client
storage and never ask the server to perform any computation on the encrypted blocks
stored on server storage. For this class of algorithms, to prove obliviousness it is sufficient
to show that for every random $\pi$ and for every $\touched$, the sequence consisting
of the sources of the download moves and of the destination of the upload moves
is independent of $\sigma$ give $\touched$ and $\pi(\touched)$.
More precisely, we define $\transmD^\Shuffle(\Enc(K,\B),\pi,\sigma,\touched)$
as the distribution of the
move transcript $\transm$ obtained from a transcript $\trans$ distributed according to
$\transD^\Shuffle(\Enc(K,\B),\pi,\sigma,\touched)$ by removing the initial encrypted blocks and the encrypted
blocks associated with upload moves.
It is not difficult to prove that if
$\transD^\Shuffle(\Enc(K,\B),\pi,\sigma,\touched)$ is independent of $\sigma$ given
$\touched$ and $\pi(\touched)$
and the encryption scheme is IND-CPA, then $\Shuffle$ is a
{\KObShuffA} algorithm.

\subsection{Probability tools}
We will use the notion of {\em negatively associated} random variables.

\begin{definition}
The random variables $X_1,\ldots,X_n$ are {\em negatively associated} if for every two disjoint
index sets, $I,J\subseteq[n]$,
$$\E[f(X_i, i\in I)\cdot    g(X_j,j\in J)]\leq
  \E[f(X_i,i\in I)]\cdot \E[g(X_j,j\in J)],
$$
for all functions $f$ and $g$ that both non-increasing or both non-decreasing.
\end{definition}
We are going to use the following property of negatively associated random variables. For a proof
see, for example, Lemma 2 of \cite{NA}.
\begin{lemma}
\label{lemma:ndf}
Let $X_1,\ldots,X_n$ be negatively associated random variables. Then,
for non-decreasing functions $f_1,\ldots,f_k$ over disjoint variable
sets $S_1,\ldots,S_k$
	$$\E\left[\prod_{i\in [k]} f_i(S_i)\right]\leq
	\prod_{i\in [k]}\E\left[f_i(S_i)\right].$$
\end{lemma}

We will also use the fact that the Balls and Bins process is negatively
associated (see Section 2.2 from~\cite{NA}).
\begin{theorem}
\label{thm:bbna}
Consider the Balls and Bins process with $m$ balls and $n$ bins.
Let $B_1,\ldots,B_n$ be the number of balls in each of the bins.
Then, $B_1,\ldots,B_n$ are negatively associated.
\end{theorem}

We use the following theorem from Queuing Theory (see~\cite{Sanders:2000} for a proof).
\begin{theorem}\label{thm:queuep}
Let $\mathsf{Q}$ be a queue with batched arrival rate $1-\epsilon$ and departure rate $1$ and let $q_t$ be the
size of the queue after $t$ batches of arrival.
Then, for all $\epsilon > 0$, $\E[e^{\epsilon q_t}]\leq 2$.
\end{theorem}

We will also use concentration inequalities over the sum of independent
binary random variables.
\begin{theorem}[Chernoff Bounds]
\label{thm:chernoff}
Let $X = X_1+\ldots+X_n$, where $X_i = 1$ with probability $p_i$
and $X_i = 0$ with probability $1-p_i$ and all $X_i$ are independent.
Let $\mu = \E[X] = p_1 + \ldots + p_n$. Then
\begin{enumerate}
\item
$\Pr[X \ge (1+\delta)\mu] \le \exp(-\frac{\delta^2 \mu}{2+\delta})$
\item
$\Pr[X \le (1-\delta)\mu] \le \exp(-\frac{\delta^2 \mu}{2})$
\end{enumerate}
\end{theorem}

\section{\ObShuffA\ with $O(\sqrt{N})$ Client Memory}
\label{sec:fullRoot}
In this section we describe $\fullC$, an {\ObShuffA} algorithm that uses $O(\sqrt{N})$
client storage except with negligible probability.
More precisely, for every $\epsilon>0$, we describe an algorithm $\fullC_\epsilon$
uses $(3+\epsilon/2)N$ server storage,
bandwidth $(4+\epsilon)N$ and, except with negligible in $N$ probability,
$\delta_\epsilon\sqrt{N}$ client storage, for some constant
$\delta_\epsilon$ that depends solely from $\epsilon$.
Whenever $\epsilon$ is clear from the context or immaterial,
we will just call the algorithm $\fullC$.

We start by describing a simple algorithm that does not work but it gives a general idea of
how we achieve shuffling using small client memory.

For permutations $(\pi,\sigma)$, the input is
an array $\source[1,\ldots,N]$ of $N$ ciphertexts stored on server storage.
An encryption of block $B_l$ is stored as $\source[\pi(l)]$, for $l=1,\ldots,N$.
The expected output is an array $\dest[1,\ldots,N]$ such position $\dest[\sigma(l)]$, contains an
encryption of $B_l$.
The $N$ indices of $\dest$ are randomly partitioned into
$q:=\sqrt{N}$ {\em destination buckets}, $\dInd_1,\ldots,\dInd_q$,
by assigning each $j\in[N]$ to a uniformly chosen destination bucket.
Then the indices of array $\source$ are partitioned into $s:=\sqrt{N}$ {\em groups}
of $N/s=\sqrt{N}$ indices
with the $j$-th group consisting of indices in the interval $[(j-1)N/s+1,\ldots,j\cdot N/s]$,
for $j=1,\ldots,s$.
On average, each bucket has $q=s$ indices and
exactly one index from each group is assigned by $\sigma$ to each destination bucket.
If this were actually the case,
then the shuffle could be easily performed as follows using only $s$ blocks of client memory.
The blocks in each group of indices of $\source$ are downloaded one at a time in client memory.
When the $j$-th group has been completely downloaded,
exactly one block is uploaded to the $j$-th position of each destination bucket.
After all groups have been processed, each destination bucket
contains all the blocks albeit in the wrong order.
This can then be fixed easily by entirely downloading each destination bucket,
one at a time, to client memory and uploading the blocks in the correct order.

Unfortunately, it is unlikely that indices will distribute nicely over destination buckets.
Algorithm $\fullC$ is similar
except that it does not expect each source group to contain exactly one block for each destination
bucket and, for the few failures,
it stores the extra blocks in a cache stored at the client's private storage
with the hope that there will never be too many extra blocks.
It turns out that, for the above statement to be true,
we need a little bit of slackness that we achieve by slightly increasing the number of partitions
of $\dest$ to $q=(1 + \epsilon/2)\sqrt{N}$, for some $\epsilon>0$.
As we shall see, the algorithm of Section~\ref{sec:fullRecursive} will adopt the same framework but for
technical reasons we will create slackness in a different way.
Let us now proceed more formally.

\subsection{$\fullC$ Description}
For $\epsilon>0$,
we next describe algorithm $\fullC_\epsilon$ for input $(\pi,\sigma)$.
Algorithm $\fullC_\epsilon$ receives as inputs the permutations $\pi$ and $\sigma$
and the {\em source array} $\source[1,\ldots,N]$ of $N$ ciphertexts
such that an encryption of block $B_l$ is stored as $\source[\pi(l)]$, for $l=1,\ldots,N$.
$\thirdC$ outputs a {\em destination array} $\dest[1,\ldots,N]$ of $N$ ciphertexts
such that an encryption of block $B_l$ is stored as $\dest[\sigma(l)]$, for $l=1,\ldots,N$.

The $N$ indices of $\source$ are partitioned by the algorithm into
$s:=\sqrt{N}$ {\em groups} $\blockSub_1,\ldots,\blockSub_s$, each
of size $N/s=\sqrt{N}$, with  $\blockSub_j$ containing indices in the interval
$[(j-1)\cdot s+1,\ldots, j\cdot s]$.
The $N$ indices of the destination array $\dest$
are randomly partitioned by the algorithm into
$q:=(1+\epsilon/2)\sqrt{N}$ {\em destination buckets},
$\dInd_1,\ldots,\dInd_q$, by assigning each $i\in [N]$ to a randomly chosen destination bucket.
A destination bucket is expected to contain $N/q\approx(1-\epsilon/2)\sqrt{N}$ locations.
In addition, for each destination bucket, the algorithm initializes
$q$ {\em temporary} arrays $\temp_1,\ldots,\temp_q$ each of size $s$ on the server
and $q$ caches $\Q_1,\ldots,\Q_{q}$ on the client.
The working of the algorithm is divided into two phases: $\Spray$ and $\Recalibrate$.

The $\Spray$ phase consists of $s$ rounds, one for each group.
In the $i$-th $\Spray$ round,
the algorithm downloads all ciphertexts in the $i$-th group $\blockSub_i$.
Each downloaded ciphertext is decrypted, thus giving a block, say $B_l$, that is
re-encrypted with fresh randomness and stored in the cache corresponding to the destination
bucket containing $\sigma(l)$, that is $B_l$'s final destination.
After all $s$ blocks of $\blockSub_i$ have been downloaded and assigned to the caches,
the algorithm uploads one block from $\Q_j$, for $j=1,\ldots,q$,
to the $i$-th position of temporary array $\temp_j$.
If a queue is empty a {\em dummy} block containing an encryption of $0$'s is uploaded instead.

Note that after the $\Spray$ phase has completed every block has been
downloaded from the source array
and some have been uploaded to a temporary array and
some are still in the caches.
Nonetheless, each temporary array contains exactly $s$ ciphertexts and
all non-dummy blocks whose encryption is in $\temp_j$ are assigned by $\sigma$ to a position in $\dInd_j$.

The $\Recalibrate$ phase has a round for each destination bucket.
In the round for destination bucket $\dInd_j$,
the algorithm downloads all blocks from temporary array $\temp_j$
in increasing order.
Each block is decrypted, dummy blocks are discarded and the remaining blocks are re-encrypted using fresh
randomness.
Now, all blocks that belong in $\dInd_j$ are in client memory and the algorithm
uploads them to the correct position in $\dInd_j$ according to $\sigma$.
We present pseudocode of the algorithm in Appendix~\ref{sec:second_code}.

\subsection{Properties of $\fullC$}
It is easy to see that $\fullC_\epsilon$ uses $(3+\epsilon/2)N$ blocks of
server memory and $(4+\epsilon)N$ blocks of bandwidth.
Next, we are going to show that, for every $\epsilon>0$ there exists $\delta$ such that
the probability that at any given time the total size of the caches exceeds
$\delta\cdot\sqrt{N}$ is negligible.

We denote by $l_{i,j}$ the size of $\Q_j$ after processing
$\blockSub_i$. Thus, we are interested in bounding
$l_i=l_{i,1}+\ldots+l_{i,q}$ for all rounds
$i=1,\ldots,s$.

\begin{lemma}
	\label{lem:cache}
For every $\epsilon > 0$, there exists $\delta$ such that $\Pr[l_i>\delta q] < e^{-q}$.
\end{lemma}
\begin{proof}
Let $X_{i,j}$ for all $i\in[s]$ and $j\in[q]$
be the number of blocks that go from $\blockSub_i$ into $\Q_j$.
For any fixed $i$, the set $X_{i,1},\ldots,X_{i,q}$ is a Balls
and Bins process with
$q$ bins.
Therefore, by Theorem~\ref{thm:bbna}, $X_{i,1},\ldots,X_{i,q}$ are negatively associated.
For any $i \ne j$, the sets of variables $X_{i,1},\ldots,X_{i,q}$
and $X_{j,1},\ldots,X_{j,q}$ are mutually independent.
By Proposition 7.1 of~\cite{NA}, the sets are also negatively associated.
Note, note that each $l_{i,j}$ is a non-decreasing function of the set of variables
$X_{1,j},\ldots,X_{i,j}$. Therefore, for any $j \ne k$, $l_{i,j}$ and
$l_{i,k}$ are non-decreasing functions over a disjoint set of negatively
associated variables.

By Markov's Inequality,
we get that $\Pr[l_i > \delta q] = \Pr[e^{\epsilon l_i} > e^{\epsilon \delta q}]<
	e^{-\epsilon \delta q} \E[e^{\epsilon l_i}]$.
For each $\Q_j$, $j = 1,\ldots,q$, the batched arrival rate
is $(N/s)/q \approx (1-\epsilon)$ and the departure
rate is $1$. So,
$$
	\E[e^{\epsilon q_i}] = \E\left[ \prod\limits_{j=1}^{q} e^{\epsilon l_{i,j}} \right]\leq
	\prod\limits_{j=1}^{q} \E[e^{\epsilon l_{i,j}}] \le 2^q.
$$
The second inequality follows from Theorem~\ref{lemma:ndf}
since $l_{i,j}$ are non-decreasing functions
over disjoint sets of negatively associated variables.
The last inequality is by Theorem~\ref{thm:queuep}.
Therefore, $\Pr[l_i > \delta q]<e^{-\left(\delta \frac{\epsilon}{1+\epsilon} - \ln 2\right)q}$.
The lemma follows when $\delta>(1+1/\epsilon) \ln 2e$.
\end{proof}
Note that, since $q=(1+\epsilon)\sqrt{N}$, the probability that any given time
the total size of the caches exceeds $\delta_\epsilon q$ is negligible in $N$.
We also remark that the $\Spray$ phase can be generalized to any two values of $s$ and $q$ such that
$s\cdot q=(1+\epsilon)N$ in which case memory $O(q)$ is used except with probability exponentially small in $q$.
This fact will be used in Section~\ref{sec:fullRecursive}.
Next we prove obliviousness.

\begin{lemma}\label{lemma:CSind}
For every $\epsilon > 0$,  $\fullC$ is an {\ObShuffA} algorithm.
\end{lemma}
\begin{proof}
It is sufficient to show that the accesses to server storage,
that is the sources of the download moves and the destinations
of the upload moves, are independent of $\sigma$, for random $\pi$.

In the $j$-th round of the $\Spray$ phase, downloads are performed from
$\blockSub_j$ and uploads have
as destination the $j$-th slot of each $\temp_i$.
Clearly these moves are independent of $\sigma$.

In the $j$-th round of the $\Recalibrate$ phase,
the downloads of $\temp_j$ occur in increasing order, independent of $\sigma$.
The uploads have as destination the entries of $\dInd_j$ in increasing order which is
clearly independent of $\sigma$.
\end{proof}

\section{\ObShuffA\ with Smaller Client Memory}
\label{sec:fullRecursive}
In this section we generalize algorithm $\fullC$ to $\thirdC$.
Specifically, for $S = \omega(\log N)$, we provide
an {\ObShuffA} algorithm that uses $O(S)$ client memory
and $O(N \log_S N)$ bandwidth.

When $\fullC$ completes the $\Spray$ phase,
all the data blocks that according to $\sigma$ belong to a location
in destination bucket $\dBuck_i$
are either on the server in $\temp_i$ or in client memory in $\Q_i$.
The $i$-th $\Recalibrate$ step then takes the blocks from each $\temp_i$ and $\Q_i$,
and arranges them so that they all end up in the right
position according to $\sigma$ in $\dBuck_i$.
The $i$-th $\Recalibrate$ step needs memory exactly equal to the size of
$\temp_i$. The key to a {\ObShuffA} that uses less client memory resides
in a $\Spray$ phase
that uses smaller memory while producing smaller $\temp_i$.
We call this new method as $\RSpray$.

\subsection{Description of $\RSpray$}
Algorithm $\RSpray$ is similar to $\Spray$ described in Section~\ref{sec:fullRoot} but
it achieves the needed slackness in a different way. Specifically,
the slackness is needed
to ensure that the arrival rate to each cache is smaller than the departure rate by at least a constant and
this is obtained by making the number $q$ of caches larger than the number of ciphertexts in an input group
by a constant factor.
$\RSpray$ instead takes a dual approach: the number of caches is equal to the number of ciphertexts
in an input group but it assumes that each group has a constant fraction of dummy ciphertexts
that need not to be added to the queue.
There is one extra subtle point.
Since we need the dummy to be uniformly distributed over the groups, $\RSpray$
partitions the input into random buckets.
Let us proceed more formally.

Algorithm $\RSpray$ receives as input {\em source array} $\inputR[1,\ldots,\inputN]$ of $\inputN$ ciphertexts
and a set $\destInd\subseteq[N]$ of $\destIndN$ {\em destination indices}.
$\inputR$ contains the encryptions of all blocks $l$ with $\sigma(l)\in\destInd$ as well as
the encryptions of some dummy blocks. Clearly, $\inputN\geq\destIndN$.
$\inputR$ is stored on the server and $\destInd$ is a private input to $\RSpray$.

$\RSpray$ is parametrized by the size $S$ of the client storage
and outputs $q:=S$ {\em temporary arrays}, $\temp_1,\ldots,\temp_q$, of ciphertexts and
a partition of set $\destInd$ into $q$ {subsets of destination indices}
$\destInd_1,\ldots,\destInd_q$.
The arrays and the subsets of the partition are linked by the following property:
if $\sigma(l)\in\destInd_j$ then one of the ciphertexts of $\temp_j$ is
an encryption of block $l$.

We next formally describe $\RSpray$.
Algorithm $\RSpray$ partitions $\destInd$ into $q$ {subsets of destination buckets},
$\destInd_1,\ldots,\destInd_q$, by assigning each index in $\destInd$
to a randomly and uniformly selected subset of the $q$.
Each subset $\destInd_j$ is associated with a temporary array $\temp_j$ stored on the server
and a cache $\Q_j$ stored on the client.
Initially, both $\temp_j$ and $\Q_j$ are empty and $\temp_j$ will grow to contain exactly
$s:= \inputN/q$ ciphertexts.
The algorithm then partitions $\inputR$ into $s$
{\em source buckets}, $\sourceB_1,\ldots,\sourceB_s$ that are stored on the server.
Each ciphertext of $\inputR$ is randomly assigned to one of the $s$ source buckets uniformly at random.

Now, just as $\Spray$, algorithm $\RSpray$ has $s$ {\em spray} rounds, one for each source bucket.
The spray round for a source bucket also terminates by uploading exactly
one ciphertext from each cache $\Q_j$ to the corresponding temporary bucket $\temp_j$.
If a cache happens to be empty, a dummy block is encrypted and uploaded.

After all spray rounds have been completed, each $\temp_j$ contains
exactly $s$ ciphertexts (as exactly one is uploaded for each source bucket)
and
we have that if an encryption of block $B_l$ was in $\inputR$ at the start of
$\RSpray$ then at the end of the spray phase an encryption of the same block
occupies a location in $\Q_j$ or $\temp_j$, where $\sigma(l)\in\destInd_j$.

Algorithm $\RSpray$ has a final {\em adjustment} phase for each $\temp_j$
in which all ciphertexts in the cache $\Q_j$ are uploaded to $\temp_j$.
This is achieved in the following way.
In the adjustment phase for $\temp_j$, each ciphertext in $\temp_j$ is downloaded and decrypted.
If decryption returns a real block (non-dummy) then the block is re-encrypted and uploaded again.
If instead a dummy block is obtained, then two cases are possible.
In the first case, $\Q_j$ is not empty;
then a ciphertext from the cache is uploaded instead.
In the second case instead $\Q_j$ is empty and a new ciphertext of a dummy block is uploaded.

If, once all adjustment phases have been completed, there is a non-empty cache $\Q_j$
then $\RSpray$ fails and aborts.

\subsubsection{Properties of $\RSpray$}
We first observe that $\RSpray$ uses bandwidth $4n$.
Indeed, in the spray phase exactly $n$ ciphertexts are downloaded from $\inputR$ to client memory
and exactly $n$ are uploaded to the temporary buckets. In the adjustment phase exactly $n$ are downloaded and
$n$ are uploaded from the temporary buckets.

Moreover, if $\RSpray$ does not abort,
we have that if an encryption of block $B_l$ was in $\inputR$ then at the end
of $\RSpray$ an encryption of $B_l$ is found in $\temp_j$ for $j$ such that
$\sigma(l)\in\destInd_j$.

We next prove that if there exists a constant $\epsilon$ such that  $d\leq(1-\epsilon)n$,
then the algorithm aborts with negligible probability. In other words, we assume that of the $n$ ciphertexts
in $\inputR$, at least an $\epsilon$ fraction consists of encryptions of dummy blocks.
We will then show that, except with negligible probability, this is the case in all calls to
$\RSpray$ of $\thirdC$.
\begin{lemma}
\label{lem:abort}
If $d\leq(1-\epsilon)n$ then $\RSpray$ aborts with probability at most $c^{-\frac{n}{S}}$
for some constant $c>1$ that only depends on $\epsilon$.
\end{lemma}
\begin{proof}
$\RSpray$ aborts when it cannot copy an encryption of each block
assigned to some $D_j$ by $\sigma$ to $\temp_j$ because $D_j$ is larger than
$\temp_j$ (note that each temporary bucket has exactly $s=n/S$ slots).
Note that $\E[|D_i|]=d/q\leq (1-\epsilon)n/S$ and it is the sum
of $d$ 0/1 independent random variables.
The lemma then follows from the Chernoff bound.
\end{proof}

We next bound the memory needed by the client to store the caches $\Q_j$.
Specifically, we show that for every $\epsilon>0$, there exists $\delta_\epsilon$ such
that, for all $\delta>\delta_\epsilon$ the probability that the total number of blocks in the caches
exceeds $\delta S$ is negligible.
As before, we let $l_{i,j}$ denote the size of $\Q_j$ after the $i$-th spray round and set
and $l_i=l_{i,1}+\ldots+l_{i,q}$.

\begin{lemma}
For every $\epsilon > 0$ and $i\in[s]$
if $d\leq(1-\epsilon)n$, there exists $\delta_\epsilon$ such that
$\Pr[l_i>\delta q] < e^{-q}$, for all $\delta>\delta_\epsilon$.
\end{lemma}
\begin{proof}
The proof proceeds as the one of Lemma~\ref{lem:cache}.
Negative associativity still holds for the $X_{i,j}$, the random 
variable of the number of blocks in $\sourceB_i$ that go into $\Q_j$, as they
have the same distribution of the Balls and Bins process with $d$ balls and $q$ bins.
By Markov's Inequality, we get that $\Pr[l_i > \delta q] = \Pr[e^{\epsilon l_i} > e^{\epsilon \delta q}]
	< e^{-\epsilon \delta q} \E[e^{\epsilon l_i}]$.
Then we observe
each source bucket has expected size $q$ and since each source bucket is randomly
chosen from a set of $n$ ciphertext at most $(1-\epsilon)n$ of which are real, each source bucket
contains on average at most $(1-\epsilon)q$ real ciphertexts.
Therefore the arrival rate at each cache of the $q$ caches is at most $(1-\epsilon)$ and departure is exactly $1$.
The proof then proceeds as in Lemma~\ref{lem:cache}.
\end{proof}

\begin{lemma}
The move transcript of $\RSpray$ is independent of $\sigma$.
\end{lemma}
\begin{proof}
The only difference between $\Spray$ and $\RSpray$ is that
how the source arrays are distributed. In $\RSpray$, each block
of $\inputR$ is assigned uniformly at random to one $\sourceB_i$
independently of $\sigma$.
The rest of the proof follows identically to $\Spray$.
\end{proof}

\subsection{Description of $\thirdC$}
We are now ready to describe algorithm $\thirdC$ which will use $\RSpray$ and $\Spray$ as subroutines
to Oblivious Shuffle with $O(S)$ client storage.
$\thirdC$ receives permutations $\pi$ and $\sigma$ a 
a {\em source array} $\source[1,\ldots,N]$ of $N$ ciphertexts
such that an encryption of block $B_l$ is stored as $\source[\pi(l)]$, for $l=1,\ldots,N$.
$\thirdC$ outputs a {\em destination array} $\dest[1,\ldots,N]$ of $N$ ciphertexts
such that an encryption of block $B_l$ is stored as $\dest[\sigma(l)]$, for $l=1,\ldots,N$.

$\thirdC$ starts by running the $\Spray$ algorithm of $\fullC$ with
parameters $s:= N/S$ and $q:=(1+\epsilon)S$.
Note, $\Spray$ will only use $O(S)$ client memory with these parameters
and results in the following:
\begin{enumerate}
\item
$q$ caches $\Q_1,\ldots,\Q_q$ on the client;
\item
$q$ temporary arrays $\temp_1,\ldots,\temp_q$ on the server;
\item
$q$ destination buckets $\dBuck_1,\ldots,\dBuck_q$ on the client
such that if $\sigma(i)\in\dBuck_j$ then
$\Q_j$ or $\temp_j$ contain an encryption of $B_i$;
\end{enumerate}
Next, for $j=1,\ldots,q$, the algorithm performs a {\em adjustment} of $\Q_j$ into $\temp_j$ as
explained above in the description of $\RSpray$.
Once adjustment has been performed, we have that for $i=1,\ldots,N$,
if $\sigma(i)\in\dBuck_j$ then $\temp_j$ contains an encryption of $B_i$.

Next, $\thirdC$ calls algorithm $\RSpray$ on each bucket $\temp_j$ until,
after $l=O(\log_S N)$ recursive calls, it obtains buckets
$\temp_{l,j}$ of ciphertexts for destination buckets $\dBuck_{l,j}$ of size smaller than $S^2$.
At this point each bucket is oblivious shuffled into the subset of $\dest$ corresponding to
the indices in the destination bucket using algorithm $\fullC$.

\subsection{Properties of $\thirdC$}
The first invocation $\Spray$ method requires $O(N)$ blocks of bandwidth.
At level $i$ of $\RSpray$ calls, there are $S^i$ calls of $\RSpray$
each on source arrays of size $O(N/S^i)$. Therefore, each level
requires $O(N)$ blocks of bandwidth and altogether $O(N\log_S N)$
blocks of bandwidth for all levels.
Finally, each of the $O(N/S^2)$ executions of $\fullC$ requires $O(S^2)$
blocks of bandwidth.
In total, $O(N\log_S N)$ blocks of bandwidth is required for $\thirdC$.
Also, note that $\thirdC$ requires $O(N)$ server memory.

The following lemma will be instrumental in proving that the abort probability of $\thirdC$ is negligible
and that $\thirdC$ uses $O(S)$ client memory.

\begin{lemma}
\label{lemma:sizeBucket}
The probability that a destination bucket of level $i$ call to $\RSpray$ has size
larger than $(1-\epsilon/2)N/S^i$ is negligible in $N$ for $S=\omega(\log N)$.
\end{lemma}
\begin{proof}
This is certainly true for the first level in which we have $n=N$ and $d=(1-\epsilon)N$.
The calls to $\RSpray$ at level $i$ of the recursion determine
a random partition of $[N]$ into $S^i$ destination buckets each of expected size $d_i=(1-\epsilon) N/S^i$.
$\RSpray$ is invoked on each destination bucket with a bucket of $n_i=N/S^i$ ciphertexts.
By applying Chernoff bound, we obtain that the probability that a level $i$ destination bucket
is larger than $(1-\epsilon/2)N/S^i$ is exponentially small in $N/S^i$.
This is negligible in $N$ since $N/S^i\geq S$ and $S=\omega(\log N)$.
\end{proof}

We are now ready to prove the following.

\begin{lemma}
\label{lemma:thirdFail}
Algorithm $\thirdC$ fails with negligible probability.
\end{lemma}
\begin{proof}
By the Union Bound we obtain that the probability that any destination bucket in the $O(N\log_S N)$ calls to
$\RSpray$ is too large remains negligible and thus, by applying Lemma~\ref{lem:abort},
we obtain that $\thirdC$ aborts with negligible probability.
\end{proof}

We now show that $\thirdC$ requires $O(S)$ client memory except with negligible probability.
\begin{lemma}
	For $S=\omega(\log N)$, $\thirdC$ requires $O(S)$ client memory except with negligible in $N$ 
	probability.
\end{lemma}
\begin{proof}
Note, that $\Spray$, $\RSpray$ and $\fullC$ all use $O(S)$ client memory except with negligible probability.
Altogether, these subroutines are called $O(N\log_S N)$ times, meaning
the probability that any single execution results in more than $O(S)$ client memory is remains negligible.
Finally, the moving of $\Q_i$ back to $\temp_i$ after $\Spray$
requires $O(1)$ extra client memory.
\end{proof}

The above lemma only works when $N/S = \omega(\log N)$ or $S \le O(N/\log N)$.
However, we note this is not an issue since when $S = O(\sqrt{N})$,
$\fullC$ should be used instead of $\thirdC$.

\begin{theorem}
$\thirdC$ is an {\ObShuffA} algorithm.
\end{theorem}
\begin{proof}
From previous sections, we have shown that the move transcripts
of $\Spray$, $\RSpray$ and $\fullC$ are independent of $\sigma$
except with negligible probability. Since there are a total of $O(N)$
calls to these three subroutines, the probability that
any subroutine is dependent on $\sigma$ is still negligible.

It remains to show the moving of $\Q_i$ into $\temp_i$ after $\Spray$ is
independent of $\sigma$. Note, the adversary sees the download and upload
to each location of $\temp_i$ in an arbitrary manner.
So, if $\thirdC$ does not fail, this process remains independent of $\sigma$.
By Lemma~\ref{lemma:thirdFail}, $\thirdC$ fails only with negligible probability.
\end{proof}

\section{\kOS\ with $O(K)$ Client Memory}
\label{sec:smallK}
In this section,
we assume that the number $K$ of touched blocks is small enough to
fit into client memory
and give a {\KObShuffA} algorithm, $\firstC$, that uses bandwidth $2N$ to
shuffle $N$ data blocks.

$\firstC$
takes as input two permutations $(\pi,\sigma)$ and the encryptions of blocks $B_1,\ldots,B_N$
in array $\source[1\ldots N]$ arranged according to $\pi$.
That is, an encryption of block $B_i$ is stored as $\source[\pi(i)]$.
In addition, the algorithm also receives $\touched$, the
set of indices of the touched blocks
as well as the set $\pi(\touched)$ of their positions in $\source$.
At the end of the algorithm, encryptions
of the same $N$ blocks will be stored in array $\dest$ arranged according to permutation $\sigma$;
that is, an encryption of block $B_i$ is stored as $\dest[{\sigma(i)}]$.
The algorithm works into two phases.

In the first phase,
algorithm $\firstC$ downloads the encryptions of the touched blocks from $\source$;
that is, the encryption of $B_i$, stored as $\source[\pi(i)]$, is downloaded for all $i\in\touched$.
Each block is decrypted, re-encrypted using fresh randomness and stored in client memory.
Once all touched blocks have been downloaded, 
algorithm $\firstC$ initializes the set $\tbDown$ of indices of data blocks
that have not been downloaded by setting $\tbDown=[N]\setminus\touched$.

The second phase consists of $N$ steps, for $i=1,\ldots,N$.
At the end of the $i$-th step, $\dest[i]$ contains an encryption of block $B_{\sigma^{-1}(i)}$.
Let us use $s$ as a shorthand for $\sigma^{-1}(i)$.
Three cases are possible.
In the first case,
an encryption of $B_s$ is not in client memory, that is $s\in\tbDown$; then the algorithm sets $r=s$.
If instead, an encryption of $B_s$ is already in client memory,
that is $s\not\in\tbDown$, and $\tbDown\ne\emptyset$,
the algorithm randomly selects $r\in\tbDown$.
In both these first two cases,
the algorithm downloads an encryption of block $B_r$ found at $\source[\pi(r)]$, decrypts it and
re-encrypts it using fresh randomness, stores it in client memory and updates $\tbDown$
by setting $\tbDown=\tbDown\setminus\{r\}$.
In the third case in which $s\not\in\tbDown$ and $\tbDown=\emptyset$, no block is downloaded.
The $i$-th step is then complete by uploading an encryption of $B_s$ to $\dest[i]$.
Note that at this point, the client memory certainly contains an encryption of $B_s$.
We present the pseudocode for this algorithm in Appendix~\ref{sec:first_code}.

In the above description, it seems like the algorithm
would require $N$ roundtrips of data between the client and the server.
However, we can easily reduce the roundtrips by grouping indexes of $\dest$ together.
Specifically, we can group indexes of $\dest$ into groups of size $O(K)$
and perform the required downloads and uploads in $O(N/K)$ roundtrips.

\subsection{Properties of $\firstC$}
Initially, the client downloads exactly $K$ blocks.
At each step,
exactly one block is uploaded and at most one is downloaded.
Therefore, client memory never exceeds $K$.
Each block is downloaded exactly once and uploaded exactly once.
So bandwidth is exactly $2N$ blocks.

\begin{theorem}
$\firstC$ is a {\KObShuffA} algorithm.
\end{theorem}
\begin{proof}
We prove the theorem by showing that the accesses of $\firstC$ to server memory are independent from $\sigma$,
for randomly chosen $\pi$, given the sets $\touched$ and $\pi(\touched)$.
This is certainly true for the downloads of the first phase as they correspond to $\pi(\touched)$.
For the second phase, we observe that at the $i$-th step an upload is made to $\dest[i]$, which is clearly independent from $\sigma$.
Regarding the downloads, we observe that the set $\tbDown$ initially contains $N-K$ elements
and it decrease by one at each step. Therefore, it will be empty for the last $K$ steps and thus
no download will be performed.
For $i\leq N-K$,
the download of the $i$-th step is from $\source[\pi(r)]$.
In the first case $r$ is a random element of $\tbDown$ and thus independent from $\sigma$;
in the second case, the download is from $\source[\pi(s)]$, with $s=\sigma^{-1}(i)$.
Since $s\not\in\touched$, for otherwise an encryption $B_s$ would have been in client memory,
the value $\pi(s)$ is independent from $\sigma$.
\end{proof}

\section{\kOS\ with Smaller Client Memory}
\label{sec:largeK}
In this section, for every $S$, we describe $\kC_S$, a \KObShuffA\ that uses $O(S)$ blocks of client memory.
Algorithm $\kC_S$ (or, simply, $\kC$)
takes as input two permutations $(\pi,\sigma)$ and the encryptions of $N$ blocks
in array $\source$ arranged according to $\pi$. In addition, the algorithm also receives
the set $\touched$ of the indices of the touched blocks
and the set $\pi(\touched)$ of their positions in $\source$.
At the end of the algorithm, encryptions
of the same $N$ blocks will be stored in array $\dest$ arranged according to permutation $\sigma$.
Algorithm $\kC_S$ can be described as consisting of the following three phases.
We let $\epsilon>0$ be a constant.

The first phase obliviously assigns the $K$ touched blocks to $q=(1+\epsilon)K/S$
{\em touched buckets}, $\tCt_1,\allowbreak\ldots,\allowbreak\tCt_q$, each consisting of $S$
ciphertexts that are encryptions of touched and dummy blocks.
Bucket $\tCt_j$, for $j=1,\ldots,q$,
is associated with the subset $\tInd_j\subseteq\sigma(\touched)$ and $\tCt_j$ contains
an encryption of touched block $B_i$ if and only if $\sigma(i)\in\tInd_j$.
This is achieved by invoking algorithm  $\thirdC$ for memory $S$ and skipping the last 
$\Recalibrate$ phase of the last invocation of $\secondC$.
The partition $(\tInd_1,\ldots,\tInd_q)$ returned is a random parition of $\sigma(\touched)$
into $q$ subsets.
The acute reader might notice that $\secondC$ does not
guarantee that each $\tCt_i$ will contain exactly $S$ ciphertexts.
We note that this can be achieved by slightly decreasing the number of
caches for a couple recursion levels of $\thirdC$.

The second phase merges the touched and the untouched blocks into $q$ buckets.
More specifically, the second phase extends the partition $(\tInd_1,\ldots,\tInd_q)$ of $\sigma(\touched)$
into a partition $(\dInd_1,\ldots,\dInd_q)$ of the set $[N]$ of the indices of array $\dest$;
that is, $\tInd_j\subseteq\dInd_j$, for $j=1,\ldots,q$.
In addition, each set of indices $\dInd_j$ is associated with a bucket of ciphertexts $\dCt_j$
that contains an encryption of every block (touched and untouched) $B_i$ such that $\sigma(i)\in\dInd_j$.
It turns out though that an approach similar to the one used in $\firstC$ would not work here and
we need a more sophisticated algorithm.
Let us see why.
Following $\firstC$, the algorithm downloads each touched bucket $\tCt_j$ to client memory
(note that each has size $S$ so it will fit into memory)
decrypt all ciphertexts, removes the dummy blocks, and re-encrypts the other blocks.
The set $\tbDown_j$ of untouched blocks of $\dInd_j$ still to be downloaded is initialized by the algorithm
as $\tbDown_j:=\sigma^{-1}(\dInd_j\setminus \tInd_j)$.
Now, the algorithm iterates through each index $i\in\dInd_j$ in increasing order.
If block $B_k$ assigned to location $i$ by $\sigma$ (that is, $k=\sigma^{-1}(i)$) is not in client memory,
the algorithms downloads its encryption stored as $\source[\pi(\sigma^{-1}(i))]$
and removes $i$ from $\tbDown_j$.
If instead it is available in client memory,
the algorithm randomly selects random index of $k\in\tbDown_j$,
removes it from $\tbDown_j$ and downloads $\source[\pi(k)]$.
When $\tbDown_j$ is empty, the algorithm does not download anything.
Unfortunately, such an algorithm is not oblivious, since the number of downloads performed
for $\dInd_j$ reveals the cardinality of $\tbDown_j$ from which the adversary obtains
the number of touched blocks associated that are assigned by $\sigma$ to $\dInd_j$.
Note that $\firstC$ does not suffer this problem as there is only one
bucket comprising all
$N$ indices. Thus, the algorithm only leaks the total number $K$ of touched blocks which
is already known to the adversary.

The merging of touched and untouched blocks is instead achieved by the following two-phase process.
As before, the algorithm has a round for each $\dInd_j$, starting with $j=1$, and the $j$-th round starts
with the algorithm downloading the $S$ ciphertexts in $\tCt_j$ and by initializing
$\tbDown_j:=\sigma^{-1}(\dInd_j \setminus \tInd_j)$.
However, unlike in the previous approach, in each round the algorithm dowloads exactly
$u_j:=|\dInd_j|-(1-\epsilon)K/q$ untouched blocks.
If more than $u_j$ untouched blocks belong to $\dInd_j$ under $\sigma$,
the algorithm fails (and we will show that this happens with negligible probability).
If instead fewer than $u_j$ untouched blocks are assigned by $\sigma$ to $\dInd_j$,
the extra downloads are used to bring to client memory untouched blocks that belong to $\dInd_q$
(or, if none is left in $\dInd_q$ to be downloaded,
blocks that belong to $\dInd_{q-1}$ are downloaded and so on).
Note that if the algorithm does not abort (that is, no more than $u_j$ touched blocks must be downloaded)
then we can continue as previously described.
Once the encryptions of all blocks have been uploaded to $\dInd_j$,
the algorithm is left with a set $\rem_j$ of extra untouched blocks
that have been downloaded during the $j$-th round. If $|\rem_j|>2\epsilon K/q$, the algorithm aborts.
Otherwise, the algorithm pads $\rem_j$ with
encryptions of dummy blocks until there are exactly $2\epsilon K/q$ blocks.
Then, $\rem_j$ is uploaded to the server.
At the end of the round, the algorithm has in client memory all blocks that are assigned by $\sigma$
to $\dInd_j$ and the round terminates by uploading the blocks in the current positions of $\dest$.
This second phase ends when
all untouched blocks have been downloaded and they have been uploaded either to
the position in $\dest$ according to $\sigma$ or are still in some $\rem_j$.
That is, for some $l$,
the algorithm has still to process $\dInd_l,\ldots,\dInd_q$.

Finally, in the third phase, the algorithm handles all touched blocks whose encryptions are
in $\rem_1,\ldots,\rem_{l-1}$ and the touched blocks whose encryptions are in $\tCt_l,\ldots,\tCt_q$.
As we shall prove the total number of remaining blocks is $c\epsilon K$ and they are shuffled
into $\dInd_l,\ldots,\dInd_q$ by using algorithm $\thirdC$ with memory $S$.

If the client has $O(\sqrt{K})$ blocks of client storage, then
we may replace $\thirdC$ with $\secondC$ above.
We refer to this construction as $\ksecondC$.

\subsection{Properties of $\kC$}
We first show that the probability that $\kC$ aborts is negligible.
In addition to the executions of $\thirdC$ failing,
$\kC$ introduces two new possible points of aborting,
when $|\dInd_j|-|\tInd_j|>u_j$ or $|\rem_j| > 2\epsilon K/q$.
We next show that when $S$ is not too small, the abort probabilityis negligible.

\begin{lemma}
If $S=\omega(\log N)$ then $\kC$ aborts with probability negligible in $N$.
\end{lemma}
\begin{proof}
The probability that $\kC$ aborts during $\thirdC$ is negligible by the result in the previous section.
Let us now compute the probability that the algorithm aborts because one of the $\rem_j$ is too large.
Note that $\rem_j=u_j-(|\dInd_j|-|\tInd_j|)=|\tInd_j|-(1-\epsilon)K/q$ and
thus if $|\rem_j|>2\epsilon K/q$ then it must be the case that
$|\tInd_j|>(1+\epsilon)K/q$.
Note, that $|\tInd_j|$ is the sum of independent binary random variables and its
expected value is $K/q$. Therefore,
by Chernoff Bounds the probability that $\rem_j$ is larger than its expected value by a constant fractions
is exponentially small in $K/q=\Theta(S)$
and thus negligible in $N$ since $S = \omega(\log N)$.

Finally, let us compute the probability
that the algorithm also aborts because $\dInd_j$ has more than $u_j$ untouched blocks.
This happens when $|\tInd_j|\leq (1-\epsilon)K/q$ which, again by Chernoff Bounds and by
the fact that $S=\omega(\log N)$, has negligible probability
as it is the probability that a sum of independent binary random variables is a constant fraction away from
its expected value.
\end{proof}

The entire algorithm requires $O(N)$ blocks of server memory
and $O(S)$ blocks of client memory.
It is clear that the first execution of $\thirdC$
requires $O(K\log_S K)$ blocks of bandwidth.
The uploading and downloading while processing destination buckets
requires at most $2N$ blocks of bandwidth.
We now show that the last execution of $\thirdC$ (or $\fullC$)
is performed over $O(\epsilon K)$ blocks.
That implies that the algorithm has a total of $2N + (1+\epsilon)O(K\log_S K)$ blocks of bandwidth.

\begin{lemma}
If $S=\omega(\log N)$ and $K\leq N/2$,
then the number of ciphertexts left before the third phase starts is at most
$$ \frac{4\cdot\epsilon}{1-\epsilon}\,K$$
except with probability negligible in $N$.
\end{lemma}
\begin{proof}
Let $\dInd_l$ be the first subset that has not been processed by the second phase of the algorithm.
Therefore the third phase receives
$$|\rem_1|+\ldots+|\rem_{l-1}|+|\tCt_l|+\ldots+|\tCt_q|$$
ciphertexts from the second phase.
We know that $|\rem_j|=2\epsilon K/q$ and therefore the $(l-1)$ $\rem_j$'s contribute
at most $2\epsilon K$ ciphertexts.
Moreover, we know that $|\tCt_j|=S$ and therefore we only need to upper bound the number
$(q-l+1)$ of touched buckets that are left for the third phase.

First observe that the $\rem_1,\ldots,\rem_{l-1}$ contain encryptions of all the untouched blocks
for $\dInd_j$ for $j=l,\ldots,q$. Therefore the number of untouched blocks in the last $q-l+1$ subsets
$\dInd_l,\ldots,\dInd_q$ is at most $2\epsilon K$.
Moreover, since each untouched block is assigned to a randomly chosen $\dInd_j$,
we have the expected number of untouched blocks in $\dInd_j$ is $(N-K)/q\geq S$.
Therefore, by Chernoff Bounds and since $S=\omega(\log N)$,
$\dInd_j$ contains at least $(1-\epsilon)(N-K)/q$ untouched
blocks except with probability negligible in $N$.
Hence, we have
$$q-l+1\leq \frac{2\epsilon K}{1-\epsilon}\cdot\frac{q}{N-K}\leq
                 2\epsilon\cdot\frac{1+\epsilon}{1-\epsilon}\frac{K}{S}.$$
\end{proof}

The above lemma assumes that $K \le N/2$. If $K > N/2$,
we can instead use $\thirdC$ without any performance loss.
Finally we prove $K$-obliviousness of $\kC$ by showing that
the transcript $\transmD^{\kC}(\B, \pi, \sigma, \touched)$ is generated
independently of $\sigma$ given $\pi$ and $\pi(\touched)$.

\begin{theorem}
        $\kC$ is a {\KObShuffA} algorithm.
\end{theorem}
\begin{proof}
We know the $\Spray$ phase and execution of $\thirdC$ are independent from previous sections.
Note, the destination buckets are revealed during the $\Recalibrate$ phase.
When we are uploading to $\dest$, round $i$ of $\Recalibrate$
uploads one block exactly to each index of $\dBuck_i$.
However, all destination buckets are generated independently of $\sigma$.
Furthermore, the cardinality of destination buckets are generated
independently of $\sigma$ implying the number of untouched blocks
downloaded each round is also independent of $\sigma$.
All untouched blocks downloaded belong to the set of indexes
$\pi([N]\setminus\touched)$, which are generated by the challenge $\calC$
independent of $\sigma$. Therefore, $\transmD^{\kC}(\B,\pi,\sigma,\touched)$
is independent of $\sigma$, given $\touched$ and $\pi(\touched)$.
\end{proof}

\section{\kOS\ with Dummy Blocks}
\label{sec:dummy}
In this section, we consider an extension of the version of a {\KObShuffA} algorithm which has
applications in fields such as Oblivious RAM constructions.
Recall, our reference scenario is a cloud storage model
with a {\em client} that wishes to outsource the storage
of $N$ data blocks of identical sizes and identified with the integers in $[N]$
to a {\em server}.
Suppose that the client also wants to upload $D$ {\em dummy} blocks identified
with the integers $N+1,\ldots,N+D$.
The values of dummy blocks are meaningless and they might be used to help mask actions
from an adversarial server.
Clearly,
we can arbitrarily pick a value for the dummy blocks and run a \KObShuffA\ algorithm on the $N+D$ blocks.
In this section we show that, at the price of having the server perform some computation,
we can design a more efficient algorithm in terms of bandwidth.
Because of the computation that must be performed by the server, the algorithm is not a move-based algorithm.

\subsection{Problem Definition}
We modify the permutation maps $\pi$ and $\sigma$ to account for dummy blocks.
Specifically, we define $\pi,\sigma:[N+D]\rightarrow([N]\cup\{\perp\})$ and, as before,
the value $\pi(j)=i\in[N]$ means that the $i$-th data block is stored in
location $j$ on the server.
Instead, if $\pi(j)=\perp$, then location $j$ on the server contains a
{\em dummy} block. Dummy blocks can be any arbitrary value but still
the server which blocks are dummy or not, since that reveals information about $\sigma$.
Note that the number, $N$, of real blocks and the number, $D$, of dummy blocks are known to an adversary and
we set $M:=N+D$.

The security game remains the same. A $K$-restricted adversary, $\calA$,
gets to know the value of $K$ indices of $\pi$. Note, these $K$ indices
might correspond to dummy blocks.
Afterwards, the challenger, $\calC$, fills in the remaining $M-K$ uniformly at random
such that each of the $N$ blocks
appears exactly once and the rest of the locations contain dummy blocks.
The crux of the security remains hiding any information about $\sigma$ from the adversarial server.

\subsection{Polynomial Interpolation}
The main tool of this section is polynomial interpolation and relies on the
property that there exists a unique degree $k-1$ polynomial that passes
through $k$ different points with distinct x-coordinates.

We will work on a field $\F_p$, where $p$ is a prime
whose bit length is at least the bit length of block encryptions (including metadata).
Suppose the client wishes to upload $n$ real blocks
$B_1,\ldots,B_n$ to locations $i_1,\ldots,i_n$
and
$d$  dummy blocks at locations $j_1,\ldots,j_d$ on the server.
The client first computes the unique degree $n-1$ polynomial $P$
that passes through the points $(i_1,\Enc(\key,B_1)),\ldots,(i_n,\Enc(\key,B_n))$.
The polynomial $P$ can be constructed by solving the
Vandermonde matrix or using Lagrangian interpolation, which is computationally
faster.
The client then sends the $n$ coefficients of $P$ to the server along
with the indices $i_1,\ldots,i_n$ and $j_1,\ldots,j_d$ in some arbitrary pre-fixed order; e.g., increasing.
The server evaluates the polynomial at the indices received and
writes the value obtained at the location specified by the indices.
Note that the $n$ coefficients need bandwidth equal to $n$ blocks to be transferred.
The obvious property that we are using here to hide which blocks are dummy and which are not
is that any subset of $n$ points of the $m=n+d$ points on which the server is asked to evaluate the polynomial $P$
would have given the same polynomial $P$. Therefore, the memory of the $n$ points corresponding to the real blocks
that the algorithm used to determine $P$ is completely lost.
Polynomial interpolation via the Vandermonde matrix was used in~\cite{partition_oram}.

\subsection{\fourthC\ Description}
In this section we describe an \KObShuffA\ algorithm, $\fourthC$, that can be used
when a constant fraction of the blocks are dummies and that uses $O(K)$ client storage.
The algorithm is parametrized by parameter $0<\epsilon<1$ and
is adapted from $\firstC$ and its design requires some extra care to make the
the polynomial interpolation technique applicable. Specifically, a naive application
of the technique might reveal the number of (or an upper bound on) the number of non-dummy blocks from
among the $K$ touched blocks.
Indeed, the actual fraction $\rho:=\frac{N}{N + D}$ of real blocks is assumed to be known
but the algorithm should not leak the fraction of real in a smaller set of blocks and,
specifically, in the set $\touched$ of the touched blocks.
However, we use the fact that if we pick any set of $L$ blocks at random,
approximately $\rho L$ blocks will be real, for $L$ large enough and,
of course, this means that approximately $(1-\rho)L$ blocks will be dummies.

We proceed to describe $\fourthC$ formally now. Similar to $\firstC$,
we download all $K$ touched blocks onto the client, that is the set
$\source[\pi(\touched)]$. Set $p = \frac{N+D}{L}$.
Now, partition the set $[N+D]$ into
$p$ subsets, $\dInd_1,\ldots,\dInd_p$. For each $d \in [N+D]$,
$d$ is assigned uniformly at random to one of $\dInd_1,\ldots,\dInd_p$.
We set $\tbDown$ to the indices of $\source$ that have yet to be downloaded
and initialize $\tbDown = [N+D]\ /\ \touched$.
We process each of the $p$ partitions, one at a time.
If $|\dInd_i| > (1+\epsilon)L$, then $\fourthC$ aborts and fails.
Fix an order of $\dInd_i$, say increasing.
For each $d \in \dInd_i$, if $\sigma^{-1}(d) \notin \tbDown$, then
an index, $r$, is chosen uniformly at random from $\tbDown$, if $\tbDown$ is non-empty.
We remove $r$ from $\tbDown$ and download the block from $\source[\pi(r)]$.
On the other hand, if $\sigma^{-1}(d) \in \tbDown$, then
$\source[\pi(\sigma^{-1}(d))]$ is downloaded and $\sigma^{-1}(d)$ is removed
from $\tbDown$.
Using $\sigma$, we may check the number of dummy and non-dummy blocks
that need to be uploaded to $\dest[\dInd_i]$. If there are more than
$(1 + \epsilon)\rho L$ non-dummy blocks, $\fourthC$ aborts and fails.
Otherwise, we apply the polynomial interpolation trick.
Specifically, we construct the $(1 + \epsilon)\rho L - 1$ degree polynomial
$f(x)$
using the points
$\{(d, \Enc(\Key, \source[\pi(\sigma^{-1}(d))])) :
d \in \dInd_i, \sigma(d) \ne \perp \}$.
If there are less than $(1 + \epsilon)\rho L$ non-dummy
blocks, we can just use dummy blocks (whose values can be chosen arbitrarily)
as points for interpolation.
The polynomial $f(x)$ is given to the server along with the set $\dInd_i$.
For all $d \in \dInd_i$, the server places
$f(d)$ into $\dest[d]$. The pseudocode of $\fourthC$ can be found
in Appendix~\ref{sec:fourth_code}.

\subsection{Properties of $\fourthC$}

\begin{theorem}
For every constant $\epsilon > 0$, $\fourthC$ uses $O(K+L)$ blocks of client
memory, $O(N+D)$ blocks of server memory and $D + (2+\epsilon) N$ blocks
of bandwidth.
\end{theorem}
\begin{proof}
Note, downloading $\source[\pi(\touched)]$ requires $K$ blocks of client memory.
During each processing phase of $\dInd_i$, an extra $L$ blocks are downloaded
onto client memory. Afterwards, exactly $L$ blocks are sent back to the server.
Therefore, at any point in time, at most $O(K+L)$ blocks are on client memory.
Only $\source$ and $\dest$ are required on the server, meaning $O(N+D)$ blocks
of server memory. For bandwidth, note that each of the $N+D$ blocks of $\source$
are downloaded exactly once. In each of the $\frac{N+D}{L}$ phases,
$(1+\epsilon)(1-\rho) L$ blocks are uploaded. Therefore, a total of
$N + D + (1+\epsilon)\rho (N + D) = D + (2+\epsilon)N$ blocks of bandwidth are required.
\end{proof}

\begin{lemma}
\label{lemma:fourth_fail}
For every constant $\epsilon > 0$ and $L = \omega(\log N)$, then
$\fourthC$ aborts with negligible probability.
\end{lemma}
\begin{proof}
$\fourthC$ fails when there exists a partition $\dInd_i$ either with
more than $(1+\epsilon)L$ indexes or more than $(1+\epsilon)\rho L$ non-dummy
blocks, that is
$|\{d \in \dInd_i : \sigma(d) \ne \perp\}| > (1+\epsilon)\rho L$.
We show both these events happen with negligible probability using
Chernoff Bounds.
Fix any partition $\dInd_i$. For any index $d \in [N+D]$, the probability
that $d \in \dInd_i$ is $\frac{1}{p} = \frac{L}{N+D}$.
We set the variable $X_d = 1$
if and only if $d \in \dInd_i$ and $X_d = 0$ otherwise.
Let $X = X_1 + \ldots + X_{N+D}$ and note that $\E[X] = L$.
By Chernoff Bounds and since $L = \omega(\log N)$,
$$
\Pr[X > (1+\epsilon)L] \le e^{-\frac{(1+\epsilon)^2}{3+\epsilon}L} = \negl(N).
$$
We further define variable $Y_d = 1$ if and only if
$d \in \dInd_i$ and $\sigma(d) \ne \perp$. Otherwise, $Y_d = 0$.
Set $Y = Y_1 + \ldots + Y_{N+D}$, which is the number of non-dummy blocks
destined for indexes in $\dInd_i$. Note that $\E[Y] = \rho L$.
By Chernoff Bounds and since $L = \omega(\log N)$,
$$
\Pr[Y > (1 + \epsilon)\rho L] \le e^{-\frac{(1+\epsilon)^2}{3+\epsilon}\rho L} = \negl(N).
$$
Therefore, the probability that $\fourthC$ fails when processing $\dInd_i$
is negligible in $N$. Finally, by Union Bound over $\dInd_1,\ldots,\dInd_p$,
we complete the proof.
\end{proof}

\begin{lemma}
\label{lemma:fourth_obv}
If $\fourthC$ does not abort, then $\fourthC$ is \kO.
\end{lemma}
\begin{proof}
$\fourthC$ and $\firstC$ only differ in that
$\fourthC$ uploads a description of a polynomial using $(1+\epsilon)\rho L$
values instead of all $L$ values like $\firstC$.
Note, picking any subset of $(1+\epsilon)\rho L$ of the $\dInd_i$ would
have resulted in the same polynomial. Therefore, it is impossible
to distinguish dummy and non-dummy blocks using the polynomial.
The rest of the proof is similar to $\firstC$.
\end{proof}

\begin{theorem}
For every $\epsilon > 0$ and $L = \omega(\log N)$, $\fourthC$ is \kO\ except
with negligible probability.
\end{theorem}
\begin{proof}
By Lemma~\ref{lemma:fourth_fail}, $\fourthC$ aborts with negligible probability.
By Lemma~\ref{lemma:fourth_obv}, if $\fourthC$ does not abort,
$\fourthC$ is \kO.
\end{proof}

\noindent
We note that when $K \ge L$, then the client storage can simply be represented as $O(K)$.

\section{Applications to Oblivious RAM}
In this section we show how to apply the constructions of the previous sections to the problem
of designing efficient ORAM.
\subsection{Problem Definition}
We formally define the Oblivious RAM problem here.
A client wishes to outsource their data to a server.
The client's data consists of $N$ blocks, each containing exactly
$B$ words. Before uploading to the server, the client
will encrypt blocks using an \indCpa\ scheme. We will implicitly assume
that before uploading a block, the client always encrypts.
Similarly, after downloading a block, the client will automatically decrypt.
Encryption ensures the data contents are not revealed to the server.
However, the patterns of accessing blocks may reveal information.
Oblivious RAM protocols protect the client from giving information
about accesses to the server. Formally, an adversarial server
cannot distinguish two patterns with the same number of accesses
in an Oblivious RAM scheme. We describe a modern variation of
the first ORAM scheme described in
\cite{Goldreich:1987:TTS:28395.28416} and show improvements by using $\firstC$.

\subsection{Original Square Root ORAM}
It is assumed that $N$ blocks are stored on the server according
some permutation $\pi$ (which could be pseudorandom like the Sometimes-Recurse Shuffle).
The permutation $\pi$ is stored on the client, hidden from the
server.
Additionally, the client initially has $\sqrt{N}$ empty slots for blocks.
To query for block $q$, the client first checks if block $q$
exists in one of the $\sqrt{N}$ client block slots.
If block $q$ is not on the client, ask the server to download the block
at location $\pi(q)$ and store it in an empty slot on the client.
Otherwise, the client asks the server for an arbitrary location that has
not previously been downloaded, which is also stored on the client.
The client can perform $\sqrt{N}$ queries (until all slots are filled),
before an oblivious shuffle occurs. In the original work,
the AKS sorting network~\cite{AKS} was used.
However, AKS is too slow for practice
due to large constants, so Batcher's Sort~\cite{Batcher}
is usually used for practical solutions.
We replace AKS with $\firstC$ with
$K = \sqrt{N}$.
The $\sqrt{N}$ downloaded blocks will act as the
revealed indices of $\touched$. We note that all revealed indices are
already on the client, so the initial download of revealed indices
can be skipped. Therefore, exactly $2N - \sqrt{N}$ blocks of bandwidth
is used by $\firstC$.
Note, for each query, we use exactly $1$ block of bandwidth.
So over $\sqrt{N}$ queries, exactly $2N$ blocks of bandwidth are required.
The amortized bandwidth of this protocol is $2\sqrt{N}$, which is at least 5x
better than any previous variant.

\begin{center}
        \begin{tabular}{ | l | l | }
        \hline
        Shuffling Algorithm & Amortized Bandwidth\\ \hline
        Batcher's Sort~\cite{Batcher} & $O(\sqrt{N} \log^2 N)$ \\ \hline
        AKS~\cite{AKS} & $O(\sqrt{N} \log N)$ \\ \hline
        MelbourneShuffle~\cite{Ohrimenko2014} & $(10 + \Theta(1)) \sqrt{N}$ \\ \hline
        \hline
        $\firstC$ & $2\sqrt{N}$ \\ \hline
        \end{tabular}
\end{center}

When we applied $\firstC$, we never showed that security remains intact.
We will now argue that the resulting construction is still an Oblivious
RAM. Consider any two access sequences of equal length. If they perform
less than $\sqrt{N}$ queries, the access sequences are clearly
indistinguishable.
Suppose that there are more than $\sqrt{N}$ queries.
Note, the $N - \sqrt{N}$ remaining untouched blocks were previously
obliviously shuffled. The adversary only knows that these blocks are
untouched, but cannot determine the plaintext identity.
However, for the $\sqrt{N}$ touched blocks, the server can identify
the exact order for which they were queried. When $\firstC$ executes,
the adversary is unable to distinguish whether any resulting block
was previously touched and/or untouched. Therefore, future queries
remain hidden from the adversary and this argument remains identical
after every execution of $\firstC$.

\subsection{Hierarchical ORAMs}
In the work of Goldreich and Ostrovsky~\cite{Goldreich:1996:SPS:233551.233553},
they presented the first polylog overhead Oblivious RAM algorithm.
Further work by Ostrovsky and Shoup~\cite{Ostrovsky:1997:PIS:258533.258606},
improved the worst case overhead.
In both schemes, an oblivious shuffle is required to permute
data blocks randomly in a manner hidden from the adversarial server.
Furthermore, a constant fraction of the data blocks are dummies whose
values can be arbitrary.
Therefore, by using $\fourthC$, we improve the hidden constants by at least
5x compared to constructions which use the MelbourneShuffle.

\section{Experiments}
\label{sec:experiments}
In this section, we empirically investigate the hidden
constants of \nameAlgoS.
We first investigate the necessary client storage of $\secondC$
for various parameters.
Also, the performance of $\fourthC$ is
compared to $\firstC$ for handling dummy blocks.

\begin{figure*}
\centering
\subfloat[Required client storage for $\secondC$.]{
\label{fig:cache_size}
\includegraphics[width=0.5\textwidth]
        {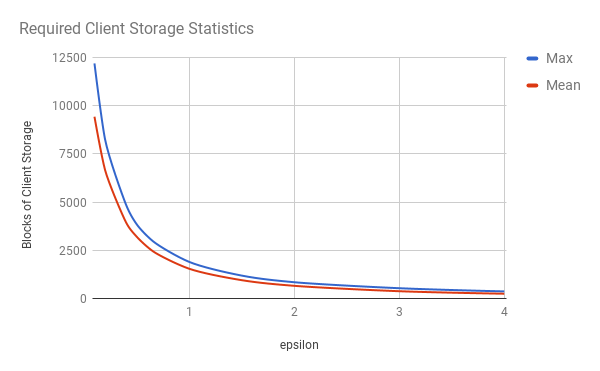}}
\subfloat[Client storage over trials for $\secondC$.]{
\label{fig:cache_trials}
\includegraphics[width=0.5\textwidth]{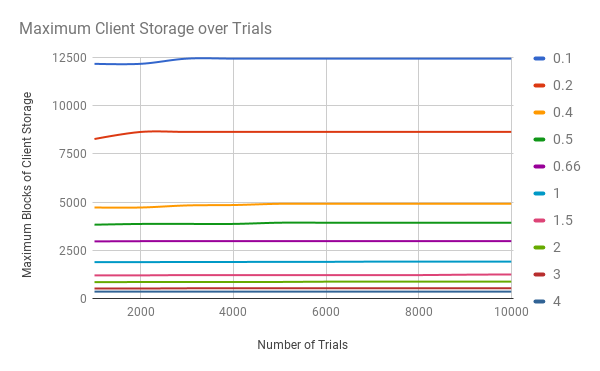}}
\caption{Client Storage of $\secondC$.}
\end{figure*}

\subsection{Client Storage of $\secondC$}

We consider multiple instantiations of $\secondC$ with various
parameters of $\epsilon$. Each instance is executed over
one million blocks of data.
From Figure~\ref{fig:cache_size}, we see that both the
max and mean cache sizes exponentially decrease as $\epsilon$
increases.
Also, the max and mean cache sizes do not differ significantly.
Furthermore, we run each instance of $\secondC$ on one million blocks
over multiple trials
and record the max cache size encountered.
It turns out the max cache size is reached fairly quickly and does not
change as the number of trials increase (see Figure~\ref{fig:cache_trials}).

\subsection{Bandwidth Comparison of $\secondC$ and MelbourneShuffle}
In these experiments, we will investigate the hidden constants of
$\secondC$ and compare them with the MelbourneShuffle when $K = N$.
Asymptotically, both algorithms require $O(\sqrt{N})$ blocks
of client storage and $O(N)$ blocks of bandwidth.
For practical use cases, the hidden constants are important.
For example, the constants affect the costs that cloud service providers
must consider for their products.
To provide a fair comparison, we will ensure to pick parameters such that
\nameAlgoS\ uses the same client storage as the Melbourne Shuffle.
The Melbourne Shuffle requires $O(\sqrt{N})$ client storage.
Therefore, we can use $\nameAlgoS_{\sqrt{K}, \epsilon}$ since $K = N$
with some small $\epsilon > 0$.

Using crude analysis of the hidden constants,
we already know that MelbourneShuffle requires $(10+c)N$
blocks of bandwidth and some $c = \Theta(1)$.
On the other hand, $\secondC$ uses only $(4+\epsilon)N$ blocks
of bandwidth, for small $\epsilon$.
In both algorithms, $c$ and $\epsilon$ are directly related with
the bandwidth as well as the hidden constants of the required
client storage.
We attempt to quantify
$c$ and $\epsilon$ for practical data sizes and a fixed number of blocks
of client storage.

Our experiments run both $\secondC$ and MelbourneShuffle
using the same input
and output permutations. Furthermore, we assume
that exactly $\sqrt{N}$ blocks of client storage are available.
It turns out that $\epsilon < 1$ is sufficient for $\secondC$,
while $c \ge 8$ is required for MelbourneShuffle.
Therefore, $\secondC$ is at least a 4x improvement over MelbourneShuffle.
A comparison of the performances can be seen in Figure~\ref{fig:melbourne}.

\begin{figure}[H]
	\centering
	\includegraphics[width=0.5\textwidth]{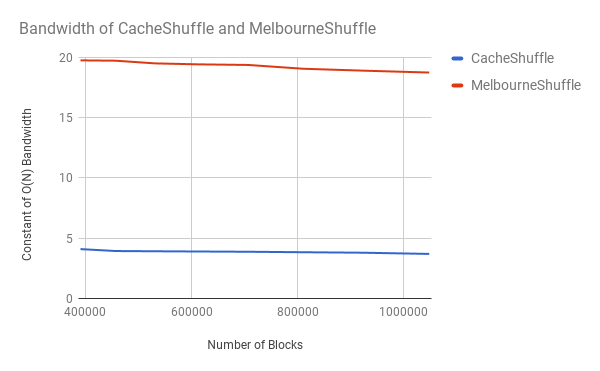}
	\caption{Comparison of $\secondC$ and MelbourneShuffle.}
	\label{fig:melbourne}
\end{figure}

Let us also consider $\sqrt{N}$-Oblivious Shuffling.
Again, MelbourneShuffle requires $(10+c)N$ blocks of bandwidth.
On the other hand, $\firstC$ requires exactly $2N$ blocks of bandwidth.
So, $\firstC$ is a 9x improvement
for practical sizes of $N$.

\subsection{Bandwidth of $\fourthC$}
We investigate the bandwidth costs of $\fourthC$ in a scenario
with dummy blocks.
For convenience, we assume that there are $N+D$ input blocks
and $N+D$ output blocks. We assume that $D=N$, that is half the blocks are dummies.
Using $\firstC$, we know that $2(N + D) = 4N$ blocks of bandwidth are required.
On the other hand, $\fourthC$ only uses $D + (2 + \epsilon)N = (3+\epsilon)N$
for some
small $\epsilon$. Using experiments, we show that
$\epsilon$ is very small for practical data sizes (see Figure~\ref{fig:dummy}).
Furthermore, as the number of blocks increase, $\epsilon$ decreases.
\begin{figure}[H]
	\centering
	\includegraphics[width=0.5\textwidth]{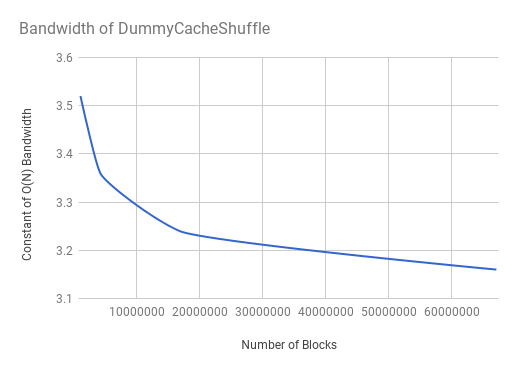}
	\caption{Bandwidth of $\fourthC$.}
	\label{fig:dummy}
\end{figure}

\bibliographystyle{abbrv}
\bibliography{biblio}

\newpage
\appendix

\section{Revisiting MelbourneShuffle}
The notion of {\em Oblivious Shuffling} was first introduced in
\cite{Ohrimenko2014}, which introduced the Melbourne Shuffle.
The Melbourne Shuffle required $O(N)$ blocks of bandwidth
and only $O(\sqrt{N})$ client memory. We show that
their security notion is an $N$-Oblivious Shuffle.

We recall the original oblivious shuffle security definition.
We show that Oblivious Shuffle is exactly $N$-Oblivious Shuffling
under the assumption that $\Enc$ is \indCpa\ secure.
Specifically, it is assumed in the original Oblivious Shuffling notion
that the adversary knows the entirety of $\pi$, the input allocation map.
\begin{definition}[Shuffle-IND]
	For challenger $\ch$ with shuffle algorithm $\Shuffle$ and
	adversary $\adv$,
	we define game $\shuffleInd_\adv^{\Shuffle}(N, \lambda)$ as follows

\begin{enumerate}
\item
	$\adv$ sends $\{B_i, \pi_i, \sigma_i\}_{i \in [l]}$ to $\ch$ where $l = \poly(\lambda)$.
\item
	$\calC$ sends $\{\pi_i(\Enc(K, B_i)), \trans_i\}_{i \in [l]}$ to $\calA$ where each $\trans_i$ is picked according to $\transD^{\Shuffle}(\pi_i, \sigma_i)$.
\item
$\calA$ submits distinct $(C_0, \rho_0, \rho'_0)$ and $(C_1, \rho_1, \rho'_1)$ to $\calA$ as the challenge.
\item
$\ch$ random selects $b\from\{0,1\}$. $\ch$ sends $\{\rho_b(\Enc(K, C_b)), \trans\}$ to $\adv$ where $\trans$ is drawn according to $\transD^{\Shuffle}(\rho_b, \rho'_b)$.
\item
Repeat Steps 1-2.
\item
$\calA$ outputs $b'$.
\end{enumerate}
Output 1 iff $b = b'$.
\end{definition}

\begin{definition}[Shuffle-IND Secure]
        Suppose that $\Shuffle$ is a shuffling
        algorithm over $N = \poly(k)$ items. Then, $\Shuffle$ is Shuffle-IND
	secure if for
        every probabilistically polynomial-time bounded adversary $\calA$
        $$
		\Pr[\shuffleInd_\adv^{\Shuffle}(N, \lambda) = 1] \le \frac{1}{2} + \negl(N).
        $$
\end{definition}

\begin{theorem}
	$\Shuffle$ is Shuffle-IND secure if and only if
	$\Shuffle$ is an Oblivious Shuffle.
\end{theorem}
\begin{proof}
	We compare the two games $\shuffleInd$ and $\osGame$ with an
	$N$-restricted adversary.
	That is, we are allowing $\adv$ to pick the entirety
	of the input permutations for $\osGame$.
	If we remove steps 1-2 and 5 from $\shuffleInd$, the games
	are identical. However, we see that $\adv$ can simulate
	steps 1-2 and 5 without the help of $\calC$ since $\Shuffle$
	is known to $\adv$. Therefore, the games are identical.
\end{proof}

\newpage\section{$\secondC$ Pseudocode}
\label{sec:second_code}\begin{algorithm}[H]
        \floatname{algorithm}{$\secondC$}
        \renewcommand{\algorithmicrequire}{\textbf{Input:}}
        \renewcommand{\algorithmicensure}{\textbf{Output:}}
        \caption{Oblivious Shuffling with $O(\sqrt{N})$ client storage.}
\begin{algorithmic}
        \REQUIRE $\pi, \sigma, \Aa, \Dd$
	\STATE Set $\ntBucket \leftarrow \sqrt{K}$ and $\ndBucket \leftarrow (1 + \epsilon/2)\sqrt{K}$.
        \STATE Set $\cnt \leftarrow 0$ and $\cnt_\ntBucket \leftarrow 1$.
        \STATE Initialize $\A_1, \ldots, \A_\ntBucket$ to be empty on the client.
        \FORALL{$i \in \K$}
                \STATE Set $\A_{\cnt_\ntBucket} \leftarrow \A_{\cnt_\ntBucket} \cup \{ i \}$.
                \STATE Increment $\cnt$ by 1.
                \IF {$\cnt = \sqrt{K}$}
                        \STATE Set $\cnt \leftarrow 0$.
                        \STATE Increment $\cnt_\ntBucket$ by 1.
                \ENDIF
        \ENDFOR
        \STATE Initialize $\dMap$ to be an empty key-value storage.
        \STATE Initialize $\D_1, \ldots, \D_{\ndBucket}$ to be empty on the client.
        \FORALL{$i \in [N]$}
                \STATE Choose $j$ uniformly at random from $[\ndBucket]$.
                \STATE Set $\D_j \leftarrow \D_j \cup \{ i \}$.
                \STATE Set $\dMap[i] \leftarrow j$.
        \ENDFOR
        \STATE Initialize $\Q_1, \ldots, \Q_{\ndBucket}$ to be empty maps on the client.
        \STATE Initialize $\T_1, \ldots, \T_{\ndBucket}$ to be empty on the server
        with $p$ empty slots for blocks each.
        \FOR{$i = 1, \ldots, \ntBucket$}
                \STATE Run $\Spray(\pi, \sigma, \dMap, \Aa, \A_i, \Q_1, \ldots, \Q_{\ndBucket}, \T_1, \ldots, \T_{\ndBucket})$.
        \ENDFOR
        \FOR{$i = 1, \ldots, \ndBucket$}
                \STATE Run $\Recalibrate(\pi, \sigma, \Aa, \Dd, \D_i, \Q_i, \T_i)$.
        \ENDFOR
\end{algorithmic}
\end{algorithm}

\begin{algorithm}[H]
        \floatname{algorithm}{$\Spray$}
        \renewcommand{\algorithmicrequire}{\textbf{Input:}}
        \renewcommand{\algorithmicensure}{\textbf{Output:}}
        \caption{The Spray phase of $\secondC$.}
\begin{algorithmic}
        \REQUIRE $\pi, \sigma, \dMap, \Aa, \A, \Q_1, \ldots, \Q_{\ndBucket}, \T_1, \ldots, \T_{\ndBucket}$
        \FORALL {$i \in \A$}
                \STATE Download $\Aa[\pi(i)]$.
                \STATE Set $B_i \leftarrow \Dec(\Key, \Aa[\pi(i)])$.
                \STATE Set $\pos_i \leftarrow \dMap[\sigma(i)]$.
                \STATE Set $\Q_{\pos_i}[\sigma(i)] \leftarrow B_i$.
        \ENDFOR
        \FOR {$i = 1, \ldots, \ndBucket$}
                \IF {$\Q_i \ne \emptyset$}
                        \STATE Remove any element $B$ from $\Q_i$.
                        \STATE Upload $\Enc(\Key, B)$ to $\T_i$.
                \ELSE
                        \STATE Upload $\Enc(\Key, \bZero)$ to $\T_i$.
                \ENDIF
        \ENDFOR
\end{algorithmic}
\end{algorithm}

\begin{algorithm}[H]
        \floatname{algorithm}{$\Recalibrate$}
        \renewcommand{\algorithmicrequire}{\textbf{Input:}}
        \renewcommand{\algorithmicensure}{\textbf{Output:}}
        \caption{The Recalibrate phase of $\secondC$.}
\begin{algorithmic}
        \REQUIRE $\pi, \sigma, \Aa, \Dd, \D, \Q, \T$
        \STATE Download $\T$ from the server.
        \FORALL {blocks $B \in \T$}
                \STATE Set $B \leftarrow \Enc(\Key, B)$.
                \STATE Set $\Q[\sigma(B.\idx)] \leftarrow B$.
        \ENDFOR
        \FORALL {$j \in \D$}
                \STATE Upload $\Enc(\Key, \Q[j])$ to $\D[j]$.
                \STATE Remove $\Q[j]$.
        \ENDFOR
\end{algorithmic}
\end{algorithm}

\newpage
\section{$\firstC$ Pseudocode}
\label{sec:first_code}\begin{algorithm}[H]
        \floatname{algorithm}{$\firstC$}
        \renewcommand{\algorithmicrequire}{\textbf{Input:}}
        \renewcommand{\algorithmicensure}{\textbf{Output:}}
        \caption{$K$-Oblivious Shuffle with $O(K)$ client storage.}
\begin{algorithmic}
        \REQUIRE $\pi, \sigma, \source, \dest, \touched$
        \STATE Initialize $\tbDown \leftarrow  [N]$.
	\STATE Initialize $\U$ to be an empty map on the client.
        \FORALL{$i \in \touched$}
                \STATE Download $B \leftarrow \source[\pi(i)]$.
		\STATE Set $\U[i] \leftarrow \Dec(\Key, B)$.
                \STATE Set $\tbDown \leftarrow \tbDown\ /\ \{\pi(i)\}$.
        \ENDFOR
        \FOR{$i = 1, \ldots, N$}
                \IF{$i \le N - K$}
                        \IF{$\sigma^{-1}(i) \notin \tbDown$}
                                \STATE Pick $j$ uniformly at random from $\tbDown$.
				\STATE Download $B \leftarrow \source[\pi(j)]$.
				\STATE Set $\U[j] \leftarrow \Dec(\Key, B)$.
                                \STATE Set $\tbDown \leftarrow \tbDown\ /\ \{j\}$.
                        \ELSE
                                \STATE Download $B \leftarrow \source[\pi(\sigma^{-1}(i))]$.
				\STATE Set $\U[\sigma^{-1}(i)] \leftarrow \Dec(\Key, B)$.
                                \STATE Set $\tbDown \leftarrow \tbDown\ /\ \{\sigma^{-1}(i)\}$.
                        \ENDIF
                \ENDIF
		\STATE Upload $\Enc(\Key, \U[\sigma^{-1}(i)])$
                to $\dest[i]$.
        \ENDFOR
\end{algorithmic}
\end{algorithm}

\newpage
\section{$\fourthC$ Pseudocode}
\label{sec:fourth_code}\begin{algorithm}[H]
        \floatname{algorithm}{$\fourthC$}
        \renewcommand{\algorithmicrequire}{\textbf{Input:}}
        \renewcommand{\algorithmicensure}{\textbf{Output:}}
        \caption{$K$-Oblivious Shuffle including dummies with $O(K)$ client storage.}
\begin{algorithmic}
        \REQUIRE $\pi, \sigma, \source, \dest, \touched$
        \STATE Initialize $\tbDown \leftarrow [M]$.
	\STATE Initialize $\U$ to be an empty map on the client.
        \FORALL{$i \in \touched$}
                \STATE Download $B \leftarrow \source[\pi(i)]$ to client memory.
		\STATE Set $\U[i] \leftarrow \Dec(\Key, B)$.
                \STATE Set $\tbDown \leftarrow \tbDown\ /\ \{i \}$.
        \ENDFOR
	\STATE Initialize $\dBuck_1,\ldots,\dBuck_p$ to be empty on client memory.
	\FOR{$i = 1,\ldots, M$}
		\STATE Pick $r$ uniformly at random from $[p]$.
		\STATE Set $\dBuck_r \leftarrow \dBuck_r \cup \{ i\}$.
	\ENDFOR
	\FOR{$i = 1,\ldots,p$}
		\IF{$|\dBuck_i| > (1+\epsilon) K$}
			\STATE Abort algorithm and fail.
		\ELSIF{$|\{d \in \dBuck_i : \sigma(d) \ne \perp\}| >
			(1+\epsilon)\rho K$}
			\STATE Abort algorithm and fail.
		\ENDIF
		\FORALL{$d \in \dBuck_i$}
			\IF {$d \in \tbDown$}
				\STATE Download $B \leftarrow \source[\pi(d)]$.
				\STATE Set $\U[d] \leftarrow \Dec(\Key, B)$.
				\STATE Set $\tbDown \leftarrow \tbDown\ /\ \{d\}$.
			\ELSE
				\IF{$\tbDown \ne \emptyset$}
					\STATE Pick $r$ uniformly at random from $\tbDown$.
					\STATE Download $B \leftarrow \source[\pi(r)]$.
					\STATE Set $\U[r] \leftarrow \Dec(\Key, B)$.
					\STATE Set $\tbDown \leftarrow \tbDown\ /\ \{r\}$.
				\ENDIF
			\ENDIF
		\ENDFOR
		\STATE Construct $f(x)$ such that $f(\sigma(d)) = \Enc(\Key, \U[d])$, for all $d \in \dBuck_i$.
		\STATE Client sends server the coefficients of $f(x)$ 
		and $\dBuck_i' \leftarrow \sigma(\dBuck_i)$.
		\FORALL{$x \in \dBuck_i'$}
			\STATE Server places $f(x)$ in location $\dest[x]$.
		\ENDFOR
	\ENDFOR
\end{algorithmic}
\end{algorithm}

\end{document}